\documentclass[journal,12pt,draftclsnofoot,onecolumn,english]{IEEEtran}
\usepackage{amsmath,bm,amssymb,blkarray,amsthm}
\usepackage{array,color,colortbl, xcolor}
\usepackage{hyperref}
\usepackage{amsmath}
\usepackage{verbatim}
\usepackage{graphicx}
\usepackage{cases}
\usepackage{setspace}
\usepackage{amsfonts}
\usepackage{indentfirst}
\usepackage{cite}
\usepackage{caption}
\usepackage{algorithm}
\usepackage{algpseudocode}
\usepackage{booktabs}
\usepackage{subfigure}
\usepackage{multirow}
\usepackage{amsthm}
\usepackage{diagbox}
\usepackage{pstricks}
\makeatletter
\newtheoremstyle{mythm}{3pt}{3pt}{}{16pt}{\bfseries}{:}{.5em}{}
\theoremstyle{mythm}
\newtheorem{theorem}{Theorem}
\newtheorem{example}{Example}

\newtheorem{definition}{Definition}
\newtheorem{remark}{Remark}

\newtheorem{proposition}{Proposition}
\newtheorem{corollary}{Corollary}
\newtheorem{lemma}{Lemma}

\hypersetup{
linkcolor=blue,
filecolor=blue,
urlcolor=red,
citecolor=green,
}
\definecolor{bluestar}{RGB}{219,238,243}
\definecolor{redstar}{RGB}{242,220,218}
\definecolor{g1}{RGB}{215,240,242}
\definecolor{g2}{RGB}{242,220,218}
\definecolor{g3}{RGB}{215,227,191}
\definecolor{g4}{RGB}{251,215,187}

\begin{document}
\title{Placement Delivery Array for Cache-Aided MIMO Systems}
\author{Yifei Huang, Kai Wan~\IEEEmembership{Member,~IEEE,}  Minquan Cheng,   Jinyan Wang, and~Giuseppe~Caire,~\IEEEmembership{Fellow,~IEEE}
	\thanks{Y. Huang, M. Cheng and J. Wang are with the Key Lab of Education Blockchain and Intelligent Technology, Ministry of Education, and also with the Guangxi Key Lab of Multi-source Information Mining $\&$ Security, Guangxi Normal University, 541004 Guilin, China (e-mail: huangyifei59@163.com, chengqinshi@hotmail.com,  wangjy612@gxnu.edu.cn). Y. Huang is also with the School of Science, Guilin University of Aerospace Technology, 541004 Guilin, China.
	}
	\thanks{K. Wan is with the School of Electronic Information and Communications, Huazhong University of Science and Technology,
		Wuhan 430074, China (e-mail: kai\_wan@hust.edu.cn).}
	\thanks{G. Caire is with the Electrical Engineering and Computer Science
			Department, Technische Universität Berlin, 10587 Berlin, Germany (e-mail:
			caire@tu-berlin.de).}
}
\maketitle 

\begin{abstract}
We consider a $(G,L,K,M,N)$ cache-aided multiple-input multiple-output (MIMO) network, where a server equipped with $L$ antennas and a library of $N$ equal-size files communicates with $K$ users, each equipped with $G$ antennas and a cache of size $M$ files, over a wireless interference channel. Each user requests an arbitrary file from the library. The goal is to design coded caching schemes that simultaneously achieve the maximum sum degrees of freedom (sum-DoF) and low subpacketization. In this paper, we first introduce a unified combinatorial structure, termed the MIMO placement delivery array (MIMO-PDA), which characterizes uncoded placement and one-shot zero-forcing delivery. By analyzing the combinatorial properties of MIMO-PDAs, we derive a sum-DoF upper bound of $\min\{KG,  Gt+G\lceil L/G \rceil\}$, where $t=KM/N$, which coincides with the optimal DoF characterization in prior work by Tehrani \emph{et al.}. Based on this upper bound,  we present two novel constructions of MIMO-PDAs that achieve the maximum sum-DoF. The first construction achieves linear subpacketization under stringent parameter constraints, while the second achieves ordered exponential subpacketization under substantially milder constraints. Theoretical analysis and numerical comparisons  demonstrate that the second construction exponentially reduces subpacketization compared to existing schemes while preserving the maximum sum-DoF.
\end{abstract}


\section{Introduction}
\label{sec: Int}
The explosive growth of multimedia services and emerging applications, such as immersive viewing and extended reality, has resulted in an unprecedented increase in mobile data traffic. This trend imposes significant stress on existing wireless infrastructures. To mitigate this challenge, \emph{coded caching}, first introduced by Maddah-Ali and Niesen~\cite{MN}, has emerged as a promising approach that exploits the idle storage at end-user devices and applies coding techniques to reduce delivery bottlenecks.

\subsection{Maddah-Ali and Niesen Coded Caching}
The Maddah-Ali and Niesen (MN) coded caching problem considers a $(K,M,N)$ caching system, where a single-antenna server maintains a library of $N$ equal-size files and communicates with $K$ users over an error-free shared link. Each user has a cache of size $M$ files, with the memory ratio defined as $\gamma = M/N.$ An $F$-division coded caching scheme consists of two phases: \emph{placement} and \emph{delivery}. In the placement phase, each file is divided into $F$ equal-size packets and certain packets are prefetched at the users, independently of their future demands. If packets are stored directly, the strategy is termed \emph{uncoded placement}; otherwise, it is called \emph{coded placement}. The parameter $F$ is referred to as the \emph{subpacketization level}. In the delivery phase, after each user requests one file, the server broadcasts coded packets such that every user can recover its desired file. The transmission load $R$ of a coded caching scheme is defined as the maximal normalized transmission amount among all the requests in the delivery phase. The coded caching gain of a coded caching scheme is defined as $g=(K(1-\gamma))/R$ where $g$ represents the number of users simultaneously served at each transmission and $K(1-\gamma)$ is the load of the conventional uncoded caching scheme. The objective of the MN coded caching problem is to minimize
the load $R$, i.e., maximize the coded caching gain $g$.

When $K\gamma$ is an integer, Maddah-Ali and Niesen~\cite{MN} proposed the well-known MN scheme, which achieves $R = (K(1-\gamma))/g,$
where $g = 1+ K\gamma$. It was later shown in~\cite{WDP,WTDPA} that the MN scheme is optimal under uncoded placement for $K \leq N$. For general parameters $(K,N)$, Yu \emph{et al.}~\cite{YMA} derived a converse on the transmission load and designed an optimal scheme under uncoded placement by eliminating redundancy in the delivery process. Some works have also explored coded placement to further reduce the load~\cite{NKS}.

Despite its optimality in load, the MN scheme suffers from prohibitively large subpacketization, given by $F = \binom{K}{K\gamma},$
which grows exponentially in $K$. Since the subpacketization directly reflects implementation complexity, reducing $F$ while retaining a reasonable load has been the focus of extensive research~\cite{KJAJA,JCLCG,YCTC,SZG,CJYT,CJTY,CJWY,CWZW,LA,YTCC,ASK}. In particular, Yan \emph{et al.}~\cite{YCTC} introduced a powerful combinatorial structure, called the \emph{placement delivery array} (PDA), that simultaneously captures uncoded placement and one-shot\footnote{Each required data packet is transmitted exactly once, either in coded or uncoded form.} delivery. Subsequent work~\cite{SDLT} demonstrated that several earlier results~\cite{KJAJA,SZG,CWZW,LA,YTCC,ASK} can be unified under the PDA framework, establishing PDA as a fundamental tool for studying coded caching with uncoded placement and reduced subpacketization under one-shot delivery.

\subsection{Multi-Antenna Coded Caching}

Following the seminal work of Maddah-Ali and Niesen~\cite{MN}, coded caching has been studied under a wide range of network topologies, including device-to-device (D2D) networks~\cite{JCM}, hierarchical structures~\cite{KNMD,JWTLCEL}, multi-server linear networks~\cite{SMK}, and wireless interference channels such as multiple-input single-output (MISO)~\cite{NMA} and multiple-input multiple-output (MIMO)~\cite{CTXL}.

The $(L,K,M,N)$ MISO coded caching system was first introduced in~\cite{NMA}, where the server is equipped with $L$ antennas. Unlike the classical MN model with a single antenna, the design objective here is to maximize the \emph{sum degrees of freedom (sum-DoF)}, i.e., the average number of users simultaneously served per time slot. Notably, when $L=1$, the sum-DoF reduces to the coded caching gain $g$ in the MN problem. Under uncoded placement and one-shot zero-forcing delivery, Naderializadeh \emph{et al.}~\cite{NMA} established the upper bound
$\text{sum-DoF} \leq \min\{K,  K\gamma + L\},$
and proposed a scheme achieving this bound with subpacketization
$F = {K \choose K\gamma}\frac{(K\gamma)!(K-K\gamma-1)!}{(K-K\gamma-L)!},
$
which is even larger than that of the MN scheme. For instance, when $(K\gamma+L)$ is divisible by $ K\gamma+1$, Mohajer \emph{et al.} \cite{MB} reduced it to ${K\choose K\gamma}$; when $K/L$ and $ (K\gamma)/L$ are integers, Lampiris \emph{et al.} \cite{LE} reduced it to $F= {K/L\choose (K\gamma)/L}$ by grouping method;  when $K\gamma\leq L$, Salehi \emph{et al.} \cite{SPSET}  constructed a class of schemes with linear subpacketization levels. 

To further explore the relationship between subpacketization and sum-DoF, Namboodiri \emph{et al.}~\cite{KES} and Yang \emph{et al.}~\cite{YWCC}  introduced the \emph{multi-antenna placement delivery array} (MAPDA), which generalizes the PDA framework to MISO systems with uncoded placement and one-shot zero-forcing delivery. They showed that existing schemes in~\cite{NMA,SCK,MB,ST,SPSET} can all be represented by suitable MAPDAs. By analyzing the associated combinatorial constraints, the same sum-DoF upper bound as in~\cite{NMA} was recovered, along with new constructions achieving the maximum sum-DoF with strictly lower subpacketization.

Coded caching in MIMO broadcast systems was first studied in~\cite{CT,CTXL}, focusing on three transmitters and three receivers. Later, Salehi \emph{et al.}~\cite{JHA} formulated a general $(G,L,K,M,N)$ MIMO coded caching system, where each transmitter has $L$ antennas and each user has $G$ antennas. When $L$ is divisible by $G$, they proposed a scheme with
$\text{sum-DoF} = G K \gamma + L,$
by decomposing the system into $G$ parallel MISO setups, enabling simultaneous transmission of multiple multicast codewords. When $L$ is not divisible by $G$, a grouping method yields
$\text{sum-DoF} = G K \gamma + \lfloor L/G \rfloor G .$
More recently, Tehrani \emph{et al.}~\cite{MMA} formulated the maximization of the sum-DoF in MIMO coded caching as an optimization problem. Based on the MN scheme, they proposed a scheme with
$\text{sum-DoF} = G K \gamma + G \lceil L/G \rceil,$
with subpacketization
$F = {K \choose K\gamma}{K-K\gamma-1 \choose \lfloor L/G \rfloor}G,$
which strictly improves upon the bound in~\cite{JHA}.

\subsection{Contributions}\label{Model}
In this paper, we develop MIMO coded caching schemes within a unified combinatorial framework, inspired by the PDA for the shared-link model and the MAPDA for the cache-aided MISO model. Our goal is to achieve the optimal sum-DoF with subpacketization as small as possible  under uncoded placement and one-shot zero-forcing delivery when $N \geq K$. The main contributions are summarized below.

$\bullet$ \textbf{MIMO placement delivery array (MIMO-PDA):}  
We introduce the \emph{MIMO-PDA}, a new combinatorial structure that simultaneously characterizes uncoded cache placement and one-shot linear delivery in cache-aided MIMO systems. Unlike the MISO setting, where each user decodes only a single stream, each MIMO user has at most $G$ degrees of freedom for decoding. This imposes more complicated requirements on the precoding matrix, which must ensure interference-free transmission of multiple packets per user at each time slot. By applying the Schwartz–Zippel lemma~\cite{DL}, we show that the MIMO-PDA guarantees the existence of such precoding matrices for any achievable sum-DoF. Moreover, the MIMO-PDA naturally generalizes existing structures: it reduces to the classical PDA when $G=L=1$, and to the MAPDA when $G=1$.

$\bullet$ \textbf{Upper bound on the sum-DoF:}  
From the combinatorial properties of the MIMO-PDA, we derive the  following upper bound on the sum-DoF
$g \leq \min\{KG,  GK\gamma + G\lceil L/G \rceil\}$. 
Furthermore, it shows that a scheme can achieve the upper bound  
if and only if each served user can decode exactly $G$ desired packets in every time slot. 
%{\red WHAT ARE THE MOTIVATION AND ADVANTAGE OF THIS SCHEME???}

$\bullet$ \textbf{Two classes of optimal MIMO-PDAs:}  
We construct two classes of MIMO-PDAs that achieve the upper bound on the sum-DoF while substantially reducing the subpacketization level. 
\begin{itemize}
    \item[-] \emph{Cyclic construction:} When the constraint $K\gamma + \lceil L/G \rceil = K$ holds, a cyclic cache placement strategy with carefully assigned integers yields MIMO-PDAs with \emph{linear} subpacketization. To the best of our knowledge, this is the first class of MIMO coded caching schemes with linear subpacketization when $G$ does not divide $L$. In order relax the constraint of this construction, the following construction which is the main result of this paper is obtained. 
    \item[-] \emph{Hybrid construction:} Using the Baranyai's Theorem \cite{Baranyai1975} and the perfect matching \cite{BJM}, we obtain our hybrid construction of  MIMO-PDA which achieves the upper bound  on the sum-DoF while significantly reduce the subpacketization. Theoretical comparisons show that the hybrid scheme achieves an exponential reduction factor of $2^{-(m-1)K_1 H(\gamma)}$ in subpacketization compared to the TST scheme in \cite{MMA} where $K=mK_1$ is the number of users for any positive integers $K_1$ and $m$, $\gamma$ the memory ratio, and $H(\cdot)$ is the entropy function.  
\end{itemize}

The remainder of this paper is organized as follows. Section II presents the system model. Section III introduces the structure of MIMO-PDA. Section IV discusses the main results and provides detailed descriptions of the proposed scheme. Section V concludes the paper, with proofs provided in the Appendices.
 
\subsection{Notations}
In this paper, we adopt the following notation. Bold capital letters, bold lowercase letters, and calligraphic letters denote arrays, vectors, and sets, respectively. For any array (or matrix) $\mathbf{P}$, $\mathbf{P}(i,j)$ denotes the entry in the $i$-th row and $j$-th column. If $a$ is not divisible by $q$, $\langle a\rangle_{q}$ denotes the least non-negative residue of $a$ modulo $q$; otherwise, $\langle a\rangle_{q}:=q$.  For any number $x$, $\lfloor x\rfloor$ denotes the largest integer not greater than $x$, $\lceil x \rceil$ denotes the smallest integer not less than $x$. $|\mathcal{K}|$ denotes the
cardinality of the set $\mathcal{K}$, and $\mathcal{K}\setminus \mathcal{T}$ is the set of elements in $\mathcal{K}$ that are not in $\mathcal{T}$. For any positive integers $a$, $b$, $t$ with $a<b$ and $t\leq b $, let $[a:b]=\{a,a+1,\ldots,b\}$, especially $[1:b]$ is denoted by $[b]$, and 
${[b]\choose t}$ is the collection of all $t$-subsets of $[b]$. We use $a|b$ to denote that $a$ divides $b$, and $a\nmid b$ to denote that $a$ does not divide $b$.  

\section{System Model }\label{sec-model}
We consider a cache-aided multiple-input multiple-output (MIMO) broadcast system comprising one server (transmitter) with $L$ antennas and $K$ users, each equipped with $G$ antennas, as illustrated in Fig.~\ref{fig-model}. The server stores a library of $N$ independent files, denoted by $\mathcal{W}\triangleq\{W_n:n\in[N]\}$, each containing $B$ i.i.d. bits. User $k\in[K]$ is equipped with a local cache that can store $MB$ bits, where $0\le M\le N$. We adopt the shorthand $\tau:=\lceil \frac{L}{G}\rceil$, $t:=\frac{KM}{N}$ and $\lfloor L/G\rfloor:=\lfloor\frac{L}{G}\rfloor$. Clearly when $L$ is divisible by $G$, we have $\tau=\lfloor L/G\rfloor$.
\begin{figure}
	\centering
	\includegraphics[width=0.8\linewidth]{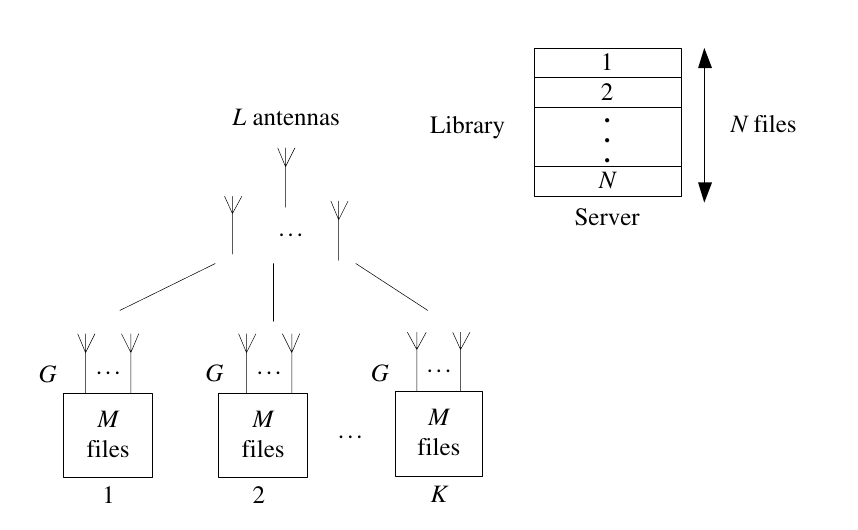}
	\caption{Multiple-input-multiple-output broadcast channel (MIMO) system of one server with $L$ antennas and $K$ users with $G$ antennas.}
	\label{fig-model}
\end{figure} 
%In each time slot, there exists a random interference 
The channel transition matrix  between the server and user $k\in[K]$ is denoted by $\mathbf{H}_{k}\in\mathbb{C}^{G\times L}$, which is completely known to the server and all users. We assume that 
the elements of channel transition matrices are uniformly i.i.d.   drawn according to some absolute continuous distribution. %, such that the probability of a measure zero
set of elements occurring is close   to $0$.

An $F$-division $(G, L, K, M, N)$  coded caching scheme consists of the following two phases.

$\bullet$ {\bf Placement phase:} The server first divides each file into $F$ non-overlapping packets with equal size, i.e., $\mathbf{w}_n=(\mathbf{w}_{n,f})_{f\in[F]}$ where $n\in[N]$, and puts some $MF$ packets into each user's memory. Let $\mathcal{Z}_k$ denote the packets cached by user $k\in[K]$. This placement strategy is called uncoded placement. In this phase, the server has no knowledge of the user's demands.

$\bullet$ {\bf Delivery phase:} Assume that each user $k\in[K]$ requests an arbitrary file $\mathbf{w}_{d_k}$ from the library where $d_k\in[N]$. Let $\mathbf{d}\triangleq (d_1,d_2,\ldots,d_K)$ represent the demand vector. Then the server transmits some coded packets by $L$ antennas to the users such that each user can decode the requested file by using its $G$ antennas and cached packets. Specifically, assume that there are $S$ blocks. At block $s$, the following packets in 
\begin{eqnarray}
	\label{eq-requested-packets}
	\mathcal{W}^{(s)}=\{\mathbf{w}_{d_{ k_{1}},f_{1,1}},\ldots,\mathbf{w}_{d_{ k_{1}},f_{1,g_{ 1}}},\ldots,\mathbf{w}_{d_{ k_{r}},f_{r,1}},\ldots,\mathbf{w}_{d_{ k_{r}},f_{{r},g_{r}}}\}
\end{eqnarray} are transmitted to the users in $\mathcal{K}^{(s)}=\{ k_{1}, k_{2},\ldots , k_{r}\}$ by the server in the following way.

The server first uses Gaussian channel coding with rate $B/\widetilde{B}=\text{log}P+o(\text{log}P)$ (bit per complex symbol) to encode each packet at each block as $\widetilde{\mathbf{w}}_{n,f}\triangleq\psi(\mathbf{w}_{n,f})$. By assuming $P$ is large enough, it can be seen that each coded packet carries one Degree-of-Freedom (DoF).\footnote{Let us first consider a point-to-point communication systems with Gaussian channel, $Y=X+Z$, where $X$ is the transmitted signal with power constraint $P,Z \sim CN(0, 1)$ is the additive independent noise, and $Y$ is the received signal. By the random Gaussian encoding, the transmission rate $R$ (i.e., the number of information bits which can be successfully transmitted per channel use) is $\text{log}P+o(\text{log}P)$. The DoF is defined as $\lim_{P \to \infty} \frac{R}{\text{log}P}$. Hence, in the above point-to-point channel, the DoF is $1$.} The whole communication process contains $S$ blocks, each of which consists of $\tilde{B}$ complex symbols (i.e., $\tilde{B}$ time slots). Then at block $s$, the server can obtain the following coded packets composed of the packets in \eqref{eq-requested-packets},
\begin{eqnarray}
\label{eq-requested-coded-packets}
\widetilde{\mathcal{W}}^{(s)}=\{\widetilde{\mathbf{w}}_{d_{ k_{1}},f_{1,1}},\ldots,\widetilde{\mathbf{w}}_{d_{ k_{1}},f_{1,g_{ 1}}},\ldots,\widetilde{\mathbf{w}}_{d_{ k_{r}},f_{r,1}},\ldots,\widetilde{\mathbf{w}}_{d_{ k_{r}},f_{{r},g_{r}}}\}.
\end{eqnarray}
Then each antenna $l\in[L]$ of the server sends the following linear combination
\begin{eqnarray}
\label{eq-x_l}
\mathbf{x}^{(s)}_l=\sum_{i\in[r]}\sum_{j\in [g_i]}v^{(s)}_{l,(i,j)}\psi(\mathbf{w}_{d_{k_i},f_{i,j}})
=\sum_{i\in[r]}\sum_{j\in [g_i]}v^{(s)}_{l,(i,j)}\widetilde{\mathbf{w}}_{d_{k_i},f_{i,j}}
\end{eqnarray}where each $v^{(s)}_{l,(i,j)}$ is a scalar complex coefficient in the precoding matrix to be designed.
So the signals transmitted by the $L$ antennas can be written as follows.
\begin{eqnarray}
\label{eq-coding-caching}
\mathbf{X}=
 \left(
  \begin{array}{c}
\mathbf{x}_1\\
\mathbf{x}_2\\
\vdots\\
\mathbf{x}_L
\end{array}
\right)%_{L\times 1}
%:=
%\left(
%\begin{array}{cccc}
%v^{(s)}_{1,(1,1)}&\cdots&v^{(s)}_{1,(r,g_r)}\\
%v^{(s)}_{2,(1,1)}&\cdots&v^{(s)}_{2,(r,g_r)}\\
%\vdots&\ddots&\vdots\\
%v^{(s)}_{L,(1,1)}&\cdots &v^{(s)}_{L,(r,g_r)}
%\end{array}
%\right)%_{L\times \sum_{i\in [g_i]}g_i}
%\left(
%\begin{array}{c}
%\widetilde{\mathbf{w}}_{k_1,f_{1,1}} \\
%\vdots\\
%\widetilde{\mathbf{w}}_{k_r,f_{r,g_{r}}}
%\end{array}
%\right)
:=
\mathbf{V}^{(s)}
\left(
\begin{array}{c}
\widetilde{\mathbf{w}}_{k_1,f_{1,1}} \\
\vdots\\
\widetilde{\mathbf{w}}_{k_r,f_{r,g_{r}}}
\end{array}
\right)
\end{eqnarray}Here $\mathbf{V}^{(s)}$ is the precoding matrix with the size $L\times \sum_{i\in [r]}g_i$. The user $k\in\mathcal{K}_s$ can receive the following transmitted signal with the help of its $G$ antennas.
\begin{eqnarray}\label{eq-y-G}
\mathbf{Y}_{k}&=\mathbf{H}^{(s)}_{k}
\mathbf{V}^{(s)}
\left(
\begin{array}{c}
\widetilde{\mathbf{w}}_{k_1,f_{1,1}} \\
\vdots\\
\widetilde{\mathbf{w}}_{k_r,f_{r,g_{r}}}
\end{array}
\right)+\epsilon_{k}
%&=\left(
%\begin{array}{cccc}
%h_{k,(1,1)}^{(s)}&\cdots&h_{k,(1,L)}^{(s)}\\
%h_{k,(2,1)}^{(s)}&\cdots&h_{k,(2,L)}^{(s)}\\
%\vdots&\ddots&\vdots\\
%h_{k,(G,1)}^{(s)}&\cdots &h_{k,(G,L)}^{(s)}
%\end{array}
%\right)
%\left(
%\begin{array}{cccc}
%v^{(s)}_{1,(1,1)}&\cdots&v^{(s)}_{1,(r,g_r)}\\
%v^{(s)}_{2,(1,1)}&\cdots&v^{(s)}_{2,(r,g_r)}\\
%\vdots&\ddots&\vdots\\
%v^{(s)}_{L,(1,1)}&\cdots &v^{(s)}_{L,(r,g_r)}
%\end{array}
%\right)
%\left(
%\begin{array}{c}
%\widetilde{\mathbf{w}}_{k_1,f_{1,1}} \\
%\vdots\\
%\widetilde{\mathbf{w}}_{k_r,f_{r,g_{r}}}
%\end{array}
%\right)+\epsilon_{k}\nonumber\\
=\mathbf{R}^{(s)}_k
\left(
\begin{array}{c}
\widetilde{\mathbf{w}}_{k_1,f_{1,1}} \\
\vdots\\
\widetilde{\mathbf{w}}_{k_r,f_{r,g_{r}}}
\end{array}
\right)
+\epsilon_{k},
\end{eqnarray}
where $\epsilon_k\sim \mathcal{CN}(0,1)$ where $k\in \mathcal{K}_s$ represents the noise vector of user $k$ at each block. We assume that $P$ is large enough, the noise can be ignored. So in the following we will omit the Gaussian channel coding. 

In this paper, we consider one-shot schemes for content delivery and decoding of $G$ packets with $G$ antennas per user and aim to maximize the sum-DoF over the linear one-shot delivery as in \cite{NMA}. More precisely, since each coded packet carries one DoF, a sum-DoF $g_{s}=\sum_{i\in[r]}g_{i}$ can be achieved in block $s$. Thus throughout the $S$ blocks, the sum-DoF can be written as $\sum_{s\in[S]}g_{s}/S$ which represents the average number of users that can be simultaneously served across all the blocks. A sum-DoF $g$ is said to be achievable if there exists a two-phase coded caching scheme in each block, where the sum-DoF of the whole system for each possible demand vector is at least $g$. Our objective is to find the maximum (or supremum) of all achievable sum-DoFs. While achieving the maximum sum-DoF, the subpacketization of the proposed scheme should be as low as possible.

\section{Multiple-input Multiple-output Placement Delivery Array}
\label{sect-MIMO-PDA}
In this section we introduce the multiple-input multiple-output placement delivery array (MIMO-PDA) which can be used to generate a MIMO coded caching scheme. %\begin{definition}[$(K,F,S)$ array]\rm
%\label{def-array}
%For any positive integers $K$, $F$ and $S$, a $(K,F,S)$ array is an $F\times K$ array $\mathbf{P}$ composed of  $``*"$  and $[S]$ such that each integer in $[S]$ occurs at most once. For any integer $s\in[S]$, the 
% sub-array indexed by $s$ is define as the $r'\times r$ array $\mathbf{P}^{(s)}=(p_{f_i,k_j}^{(s)})_{i\in[r'],j\in[r]}$ which includes the rows and columns of $\mathbf{P}$ containing $s$, where $r'$ and $r$ are nonnegative integers.
%\begin{itemize}
%	\item {\bf Support set:} For a row $f_i$ ($i\in[r']$) of $\mathbf{P}^{(s)}$, its \emph{support set} is  
%	$\mathcal{P}^{(s)}_{f_i}=\{j\in [r] \mid p_{f_i,k_j}^{(s)}\in[S]\}$, i.e, the set of columns in $\mathbf{P}^{(s)}$ whose entry in row $f_i$ is not $``*"$;
%	\item \textbf{Consistency number:}  
%\end{itemize}
% \hfill $\square$
%\end{definition}
\begin{definition}[MIMO-PDA]\rm
	\label{def-MIMO-PDA}
	For any positive integers $G$, $L$, $K$, $F$, $Z$, and $S$, an $F\times K$ array $\mathbf{P}$ composed of  $``*"$  and $[S]$ is called $(G,L,K,F,Z,S)$ multiple-input multiple-output placement delivery array (MIMO-PDA) if it satisfies the following conditions:
	\begin{itemize}
		\item [C$1$.] The symbol $``*"$ appears $Z$ times in each column;
		\item [C$2$.] Each integer occurs at least once in the array;
		\item[C$3$.] Each integer $s\in[S]$ appears at most $G$ times in each column;
		\item[C$4$.] For any integer $s\in[S]$, let $\mathbf{P}^{(s)}=(\mathbf{P}^{(s)}(f_i,k_j))_{i\in[r'],j\in[r]}$ be the subarray of $\mathbf{P}$ including the rows and columns containing $s$, where $r'$ and $r$ are nonnegative integers. Define $\mathcal{P}^{(s)}_{f_i}=\{j\in [r] \mid \mathbf{P}^{(s)}(f_i,k_j)\in[S]\}$ for each $i \in [r']$, representing the set of columns of $\mathbf{P}^{(s)}$ whose entry at row $f_i$ is not $``*"$. Then the following two conditions should hold: 
		\begin{itemize}
			\item[a)] The number of integer entries in each row of $\mathbf{P}^{(s)}$ is less than or equal to $\lceil\frac{L}{G}\rceil$, i.e., $|\mathcal{P}^{(s)}_{f}|\leq \lceil\frac{L}{G}\rceil$ for each $f \in [r']$;
			\item[b)] Let $\rho=\langle L\rangle_G$. For any $\rho+1$ different integers $f_1$, $f_2$, $\ldots$, $f_{\rho+1}\in [r']$ satisfying $\mathbf{P}^{(s)}(f_1,k)=\cdots=\mathbf{P}^{(s)}(f_{\rho+1},k)=s$ for some $k$,  there exist at least different two integers $i,i'\in [\rho+1]$ such that $\mathcal{P}^{(s)}_{f_{i}}\neq \mathcal{P}^{(s)}_{f_{i'}}$ always holds.
		\end{itemize}
	\end{itemize}
	\hfill $\square$
\end{definition}

In Definition \ref{def-MIMO-PDA}, if $G$ divides $L$ then $\langle L\rangle_G=G$, and condition C$4$-a is automatically satisfied by condition C$3$ (since an integer can appear in a column at most $G$ times). 
%The array $\mathbf{P}$ is \emph{$g$-regular}, denoted as a $g$-$(G,L,K,F,Z,S)$ MIMO-PDA, if each integer appears exactly $g$ times. 
When $L=G=1$, Definition~\ref{def-MIMO-PDA} specializes to the classical PDA in~\cite{YCTC}; when $G=1$, it coincides with the MA-PDA in~\cite{YWCC}. Hence, MIMO-PDA strictly generalizes both PDA and MA-PDA.  Given a $(G,L,K,F,Z,S)$ MIMO-PDA $\mathbf{P}$, columns to represent the user indexes and rows represent the packet indexes. If $p_{j,k}=*$, then user $k$ caches the $j$-th packet of every file.  Let us then explain the conditions:

$\bullet$ Conditions C$1$ and C$2$ are as the original PDA conditions, where condition C$1$ ensures that all users have equal cache size with memory ratio $\gamma=\frac{Z}{F}$, and condition C$2$ ensures that each integer represents a multicast coded packet in one time slot; the occurrence number of each integer represents the multicast gain in transmitting these corresponding packets together.   %If $p_{j,k}=s\in[S]$, the $j$-th packet is not cached by user $k$; during block $s$, the server broadcasts $L$ multicast signals (linear combinations of the packets indexed by $s$)

$\bullet$  Condition C$3$ ensures that each user recovers at most $G$ intended packets in one time slot, since the number of receive antennas per user is $G$. 

$\bullet$ Condition C$4$ which ensures   zero-forcing the interferences, is the most non-trivial condition in the proposed structure. First we explain condition C$4$-a by focusing on one $\mathbf{P}^{(s)}$ for some $s\in [S]$. The number of integer entries in each row of $\mathbf{P}^{(s)}$ is assumed to be $b$; so each packet in this row is an interference to $b-1$ users each of which has $G$ receive antennas. We aim to cancel the interference of this packet by using zero-forcing method at the transmitter for the total $(b-1)G$ antennas of these $b-1$ users. Hence, we need $L\geq (b-1)G+1$; in other words, $b\leq \lceil\frac{L}{G}\rceil$. Then we explain condition C$4$-b. 
% $\rho=\langle L\rangle_G=L-\lfloor\frac{L-1}{G}\rfloor G= L- (\lceil\frac{L}{G}\rceil-1) G$ 
 $\rho=\langle L\rangle_G= L- (\lceil\frac{L}{G}\rceil-1) G$ 
represents the number of linearly independent vectors which can zero-force one packet at $\lceil\frac{L}{G}\rceil-1$ users. So if we cancel $\rho+1$ packets (desired by one user) at the same $\lceil\frac{L}{G}\rceil-1$ users, these $\rho+1$ packets cannot be transmitted  with linearly independent beamforming vectors, which contradicts the decodability; thus  condition C$4$-b should hold.

Later (in the proof of Theorem~\ref{th-Fundamental}) we will prove that once the above conditions are satisfied, we can obtain an achievable scheme with high probability.

Let $g_s$ be the number of occurrences of $s\in[S]$ in $\mathbf{P}$.
 The achieved sum degrees of freedom (sum-DoF) is $\frac{\sum_{s\in [S]}g_s}{S}$. Consequently, we obtain an $F$-division $(G,L,K,M,N)$ MIMO coded caching scheme with memory size $M=\frac{ZN}{F}$ as formalized below.
\begin{theorem}\rm
\label{th-Fundamental}Given a $g$-$(G,L,K,F,Z,S)$ MIMO-PDA, there  exists an $F$-division $(G,L,K,M,N)$ coded caching scheme for MIMO network with the memory ratio $\gamma=\frac{Z}{F}$ and sum-DoF $g$.
\hfill $\square$
\end{theorem}
\begin{proof}

Given a $(G,L,K,F,Z,S)$ MIMO-PDA $\mathbf{P}$, we can obtain an $F$-division $(G,L,K,M,N)$ MIMO coded caching scheme with memory size $M=\frac{ZN}{F}$ as follows.

$\bullet$ {\bf Placement phase:}  Each file $\mathbf{w}_n$ is divided into $F$ packets of equal size, i.e., $\mathbf{w}_{n}=(\mathbf{w}_{n,f})_{f\in [F]}$. For each $k\in[K]$, user $k$ caches the packets indicated by the stars in the $k$-th column of $\mathbf{P}$,
\begin{eqnarray}
\label{eq-caching-content}
\mathcal{Z}_k=\{\mathbf{w}_{n,f}\ \mid\ p_{f,k}=*,\ f\in [F],\ n\in [N]\}.
\end{eqnarray}
Then each user caches $M=ZN/F$ files.

$\bullet$ {\bf Delivery phase:} Given a request vector ${\bf d}$, assume that the integer $s\in[S]$ appears at entries
$$\mathbf{P}^{(s)}(f_{1,1},k_1), \ldots,\mathbf{P}^{(s)}(f_{1,g_1},k_1),\mathbf{P}^{(s)}(f_{2,1},k_2), \ldots, \mathbf{P}^{(s)}(f_{r,g_r},k_r).$$
By Condition C$3$ in Definition~\ref{def-MIMO-PDA}, $g_i\leq G$ for each $i\in[r]$. Without loss of generality, order the columns so that $k_1<k_2<\cdots<k_r$. From \eqref{eq-caching-content}, user $k_i$ does not cache its requested packets $\mathbf{w}_{d_{k_i},f_{i,1}}$, $\ldots$, $\mathbf{w}_{d_{k_i},f_{i,g_{i}}}$ because $\mathbf{P}(f_{i,1},k_{i})=\ldots=\mathbf{P}(f_{i,g_i},k_{i})\neq *$. At block $s$, define the set of served users and the stacked vector of their requested packets as
\begin{equation}
\label{eq-packet-user-time-s}
\mathcal{K}_s=\{k_1, k_2, \ldots, k_r\},\ \
\mathbf{W}^{(s)}=\left(
\begin{array}{c}
\mathbf{w}_{d_{k_1},f_{1,1}} \\
\vdots\\
\mathbf{w}_{d_{k_r},f_{r,g_{r}}}
\end{array}
\right).
\end{equation}
According to \eqref{eq-coding-caching}, the server transmits $\mathbf{w}_{d_{k_i},f_{i,1}},\ldots,\mathbf{w}_{d_{k_i},f_{i,g_i}}$ to user $k_i$ for each $i\in[r]$. Consequently, user $k_i$ observes the $G$-dimensional signal ${\bf y}_{k_i}$ described in \eqref{eq-y-G}. We now establish that, under Conditions C$3$ and C$4$ of Definition~\ref{def-MIMO-PDA}, there exists a precoding matrix  $\mathbf{V}^{(s)}$ 
 enabling every $k_i\in\mathcal{K}_s$ to decode its $g_i$ desired packets $\mathbf{w}_{d_{k_i},f_{i,1}},\ldots,\mathbf{w}_{d_{k_i},f_{i,g_{i}}}$. 
 
For convenience, use $f_{i,j}$ and $k_i$ where $i\in[r]$ and $j\in[g_i]$, to denote the row and column labels of $\mathbf{P}^{(s)}$, respectively. For each row label $f_{i,j}$, denote the set of users who does not cache $\mathbf{w}_{d_{k_i},f_{i,j}}$ by  
\begin{eqnarray}
\label{eq-not-cached}
\mathcal{C}_{i,j}=\{k_{i'}\in \mathcal{K}_s\ \mid \mathbf{P}(f_{i,j},k_{i'})\in [S],\ i'\in [r]\}.
\end{eqnarray}Clearly $k_i\in \mathcal{C}_{i,j}$. Let $l_{i,j}=|\mathcal{C}_{i,j}|$. By C$4$-b, $l_{i,j}\leq\lceil\frac{L}{G}\rceil$. The set $\mathcal{C}_{i,j}\setminus \{k_i\}$ has $l_{i,j}-1$ elements; denote it by $\mathcal{C}_{i,j}'=\{k_{i,j,1},k_{i,j,2},\ldots,k_{i,j,l_{i,j}-1}\}$.

For each $l\in[l_{i,j}-1]$, let $\mathbf{H}^{(s)}_{k_{i,j,l}}$ be channel matrix, which has $G$ rows and $L$ columns, from the server to user $k_{i,j,l}$. Stack these to form
\begin{align}
\label{eq-channel-M}
\bar{\mathbf{H}}^{(s)}_{i,j}=\left(
\begin{matrix}
\mathbf{H}^{(s)}_{k_{i,j,1}}\\
\mathbf{H}^{(s)}_{k_{i,j,2}} \\
\vdots \\
\mathbf{H}^{(s)}_{k_{i,j,l_{i,j}-1}} \\
\end{matrix}
\right),
\end{align} which has $L_i= (l_{i,j}-1)G\le \frac{L-\rho}{G}\times G=L-\rho$ rows, where $\rho=L-\lfloor\frac{L-1}{G}\rfloor G$ (Definition~\ref{def-MIMO-PDA}). By our hypothesis that the entries of each channel matrix are i.i.d. under a continuous distribution, $\bar{\mathbf{H}}^{(s)}_{i,j}$ is full row rank with high probability, i.e., $\text{rank}(\bar{\mathbf{H}}^{(s)}_{i,j})=L_i$ since $L_i\leq L$.

For simplicity, the columns of $\mathbf{V}^{(s)}$ are indexed by $(i,j)$, where $(i,j)$ corresponds to the required packet $\mathbf{w}_{d_{k_i},f_{i,j}}$; that is, the column labels are $(1,1)$, $(1,2)$, $\ldots$, $(1,g_1)$, $\ldots$, $(r,g_r)$. To ensure packet $\mathbf{w}_{d_{k_i},f_{i,j}}$ is nulled at $k_{i,l}$, thereby  canceling interference at unintended receivers, we impose
\begin{align}
\label{eq-unique-vector}
\bar{\mathbf{H}}^{(s)}_{i,j}\mathbf{v}^{(s)}_{i,j}
=\left(
\begin{array}{c}
	0 \\
	0 \\
	\vdots \\
	0 \\
\end{array}
\right)_{(L-\rho)\times 1}
\end{align}where $\mathbf{v}^{(s)}_{i,j}$ is the $(i,j)$-th column of $\mathbf{V}^{(s)}$. By linear algebra, there always exists a nonzero column vector $\mathbf{v}^{(s)}_{i,j}$. Setting $\mathbf{V}^{(s)}_{i}=(\mathbf{v}^{(s)}_{i,1}, \mathbf{v}^{(s)}_{i,2},\ldots,\mathbf{v}^{(s)}_{i,g_i})$, the pre-matrix can be represented as $\mathbf{V}^{(s)}=(\mathbf{V}^{(s)}_{1},\mathbf{V}^{(s)}_{2},\ldots,\mathbf{V}^{(s)}_{r})$. 

In the sequel, we show that user $k_i$ can decode all the packets intended for it within $\mathbf{W}^{(s)}$ using the  pre-matrix $\mathbf{V}^{(s)}$ constructed from \eqref{eq-unique-vector}. We claim that the rank of $\mathbf{V}^{(s)}_i$ is $g_i$ with high probability. The detailed proof is included in Appendix \ref{app-proof-rank}.  Then from \eqref{eq-y-G} user $k_i$ can decode all of its requested packets from $\mathbf{W}^{(s)}$ by the received signal ${\bf y}_{k_i}$ since $\mathbf{H}^{(s)}_{k_i}$ is a full row rank matrix. In addition, each user has the same number of uncached packets per file, specifically $F - Z$. So the  sum-DoF in the whole procedure is $K(F-Z)/S$.
\end{proof}

Let us take the following example to further illustrate our statement in Theorem \ref{th-Fundamental}.
\begin{example}\label{ex-1}
The following array $\mathbf{P}$ is a $g$-$(G,L,K,F,Z,S)=6$-$(2,3,4,8,1,4)$ MIMO-PDA,
\begin{eqnarray}
\label{ex-MATPDA}
\mathbf{P}=\left(\begin{array}{cccc}
*&1 &1&3 \\
1&*&1&2\\
1 &1&*&2\\
3 &2&2&*\\
*&4&3&4\\
4&*&2&4\\
3&2&*&3\\
4&4&3&*
\end{array}\right).
\end{eqnarray}
Notice that each column has $Z=2$ stars and each integer appears at most $G=2$ times in each column. This implies that C$1$ and C$3$ hold. When $s=1$, we have the subarray $\mathbf{X} $ which includes the rows and columns containing $1$ as follows
$$\mathbf{X} = \left(\begin{array}{ccc}
* & 1 & 1  \\
1 & * & 1  \\
1 & 1 & *  \\
\end{array}\right).$$
Clearly each row of $\mathbf{X}$ contains $\lceil\frac{L}{G}\rceil= 2$ integer entries. So C$4$ holds. Now let us use $\mathbf{P}$ in \eqref{ex-MATPDA} to generate a $(G=2,L=3,K=4,M=1,N=4)$ MIMO coded caching scheme as follows.
\begin{itemize}
\item \textbf{Placement phase:} By the proof of Theorem \ref{th-Fundamental}, each file is divided into $8$ packets, i.e., $\mathbf{w}_n=(\mathbf{w}_{n,1},\mathbf{w}_{n,2},\ldots,\mathbf{w}_{n,8})$ for each $n\in [4]$. Then all the users cache the following packets.
\begin{align*}
\mathcal{Z}_{1} &= \{\mathbf{w}_{i,1},\mathbf{w}_{i,5} \mid i\in[4]\}, \ \
\mathcal{Z}_{2} = \{\mathbf{w}_{i,2},\mathbf{w}_{i,6} \mid i\in[4]\}, \\
\mathcal{Z}_{3} &= \{\mathbf{w}_{i,3},\mathbf{w}_{i,7} \mid i\in[4]\}, \ \
\mathcal{Z}_{4} = \{\mathbf{w}_{i,4},\mathbf{w}_{i,8} \mid i\in[4]\}.
\end{align*}Clearly each user caches packets of size $\frac{2N}{8}=1$ file.
\item \textbf{Delivery phase:} Assume that the request vector is $\mathbf{d} =(1,2,3,4)$. There are $S=4$ blocks in the whole communication process. We first consider the communication process in block
$s=1$, and other blocks can be treated similarly. Since $\mathbf{P}(2,1)=\mathbf{P}(3,1)=\mathbf{P}(1,2)=\mathbf{P}(3,2)=\mathbf{P}(1,3)=\mathbf{P}(2,3)=1$, the served users and their required packets are
\begin{align*}
\mathcal{K}_{1}=\{ 1,2,3\}, \mathbf{W}^{(1)}=(\mathbf{w}_{1,2}, \mathbf{w}_{1,3},\mathbf{w}_{2,1}, \mathbf{w}_{2,3}, \mathbf{w}_{3,1},
\mathbf{w}_{3,3})^{T}
\end{align*}
respectively. Note that users $1$, $2$ and $3$ can recover the packet pairs  $(\mathbf{w}_{2,1},\mathbf{w}_{3,1})$, $(\mathbf{w}_{1,2},\mathbf{w}_{3,2})$ and $(\mathbf{w}_{1,3},\mathbf{w}_{2,3})$ respectively. And they require the packet pairs
$(\mathbf{w}_{1,2},\mathbf{w}_{1,3})$, $(\mathbf{w}_{2,1},\mathbf{w}_{2,3})$ and $(\mathbf{w}_{3,1},\mathbf{w}_{3,2})$ respectively.  Furthermore, $\mathbf{w}_{1,2}, \mathbf{w}_{1,3}, \mathbf{w}_{2,1}, \mathbf{w}_{2,3}, \mathbf{w}_{3,1}$ and $\mathbf{w}_{3,3}$ are
the interferences of users $\mathcal{C}_{1,1}'(1)=3, \mathcal{C}_{1,2}'(1)=2, \mathcal{C}_{2,1}'(1)=3, \mathcal{C}_{2,2}'(1)=1, \mathcal{C}_{3,1}'(1)=2$ and $\mathcal{C}_{3,2}'(1)=1$ respectively. We have $\bar{\mathbf{H}}_{1,1}^{(1)}=\bar{\mathbf{H}}_{2,1}^{(1)}=\mathbf{H}_{3}^{(1)}$, $\bar{\mathbf{H}}_{1,2}^{(1)}=\bar{\mathbf{H}}_{1,3}^{(1)}=\mathbf{H}_{2}^{(1)}$ and $\bar{\mathbf{H}}_{2,2}^{(1)}=\bar{\mathbf{H}}_{3,2}^{(1)}=\mathbf{H}_{1}^{(1)}$.
Let $\mathbf{v}_{1,1}^{(1)},\mathbf{v}_{1,2}^{(1)},\mathbf{v}_{2,1}^{(1)},\mathbf{v}_{2,2}^{(1)},\mathbf{v}_{3,1}^{(1)}$, and $\mathbf{v}_{3,2}^{(1)}$ be the right null space solution vectors of $\mathbf{H}_{3}^{(1)}, \mathbf{H}_{2}^{(1)}, \mathbf{H}_{3}^{(1)}, \mathbf{H}_{1}^{(1)}, \mathbf{H}_{2}^{(1)}$ and $\mathbf{H}_{1}^{(1)}$, respectively. That is,  $\mathbf{H}_{3}^{(1)}\mathbf{v}_{(1,1)}^{(1)}=\mathbf{H}_{2}^{(1)}\mathbf{v}_{(1,2)}^{(1)}=\mathbf{H}_{3}^{(1)}\mathbf{v}_{(2,1)}^{(1)}=\mathbf{H}_{1}^{(1)}\mathbf{v}_{(2,2)}^{(1)}=\mathbf{H}_{2}^{(1)}\mathbf{v}_{(3,1)}^{(1)}=\mathbf{H}_{1}^{(1)}\mathbf{v}_{(3,2)}^{(1)}={\bf 0}$. For instance, we assume that
\begin{align*}
\mathbf{H}_{1}^{(1)}=\left(\begin{array}{ccc}
1 & 1 & 1 \\
1 & 2 & 2
\end{array}\right),\ \ \
\mathbf{H}_{2}^{(1)}=\left(\begin{array}{ccc}
1 & 4 & 3 \\
1 & 8 & 4
\end{array}\right),\ \ \
\mathbf{H}_{3}^{(1)}=\left(\begin{array}{ccc}
1 & 16 & 5 \\
1 & 32 & 6
\end{array}\right).
\end{align*}
Then we have
\begin{align*}
&\mathbf{v}_{1,1}^{(1)}=\left(\begin{array}{c}
64  \\
1  \\
-16
\end{array}\right),
\mathbf{v}_{1,2}^{(1)}=\left(\begin{array}{c}
8  \\
1  \\
-4
\end{array}\right),
\mathbf{v}_{2,1}^{(1)}=\left(\begin{array}{c}
-64  \\
-1  \\
16
\end{array}\right),\\
&\mathbf{v}_{2,2}^{(1)}=\left(\begin{array}{c}
0  \\
1  \\
-1
\end{array}\right),
\mathbf{v}_{3,1}^{(1)}=\left(\begin{array}{c}
-8  \\
-1  \\
4
\end{array}\right),
\mathbf{v}_{3,2}^{(1)}=\left(\begin{array}{c}
0  \\
-1  \\
1
\end{array}\right),
\end{align*}
as their right null space solution vectors. User $1$ can receive the following signal.
\begin{equation}
	\begin{aligned}
		{\bf y}_{1}(1)=&\mathbf{H}_{1}^{(1)}(\mathbf{v}_{1,1}^{(1)},\mathbf{v}_{1,2}^{(1)},\mathbf{v}_{2,1}^{(1)},\mathbf{v}_{2,2}^{(1)},\mathbf{v}_{3,1}^{(1)},  \mathbf{v}_{3,2}^{(1)})(\mathbf{w}_{1,2}, \mathbf{w}_{1,3}, \mathbf{w}_{2,1}, \mathbf{w}_{2,3}, \mathbf{w}_{3,1},\mathbf{w}_{3,3})^{T}\\
		=&\mathbf{H}_{1}^{(1)}\mathbf{v}_{1,1}^{(1)}\mathbf{w}_{1,2}+\mathbf{H}_{1}^{(1)}\mathbf{v}_{1,2}^{(1)}\mathbf{w}_{1,3}+\mathbf{H}_{1}^{(1)}\mathbf{v}_{2,1}^{(1)}\mathbf{w}_{2,1}+\mathbf{H}_{1}^{(1)}\mathbf{v}_{2,2}^{(1)}\mathbf{w}_{2,3} \\ 
		&+\mathbf{H}_{1}^{(1)}\mathbf{v}_{3,1}^{(1)}\mathbf{w}_{3,1}+\mathbf{H}_{1}^{(1)}\mathbf{v}_{3,2}^{(1)}\mathbf{w}_{3,2}\\
		=&\mathbf{H}_{1}^{(1)}\mathbf{v}_{1,1}^{(1)}\mathbf{w}_{1,2}+\mathbf{H}_{1}^{(1)}\mathbf{v}_{1,2}^{(1)}\mathbf{w}_{1,3}+(\mathbf{H}_{1}^{(1)}\mathbf{v}_{2,1}^{(1)}\mathbf{w}_{2,1}+\mathbf{H}_{1}^{(1)}\mathbf{v}_{3,1}^{(1)}\mathbf{w}_{3,1})\\
		=&\left(\begin{array}{c}
			49  \\
			34
		\end{array}\right)\mathbf{w}_{1,2}+\left(\begin{array}{c}
			5  \\
			2
		\end{array}\right)\mathbf{w}_{1,3}+\left(\left(\begin{array}{c}
			-49  \\
			-34
		\end{array}\right)\mathbf{w}_{2,1}+\left(\begin{array}{c}
			-5  \\
			-2
		\end{array}\right)\mathbf{w}_{3,1}.
		\label{delivery}\right)
	\end{aligned}
\end{equation}

Note that $(\mathbf{H}_{1}^{(1)}\mathbf{v}_{2,1}^{(1)}\mathbf{w}_{2,1}+\mathbf{H}_{1}^{(1)}\mathbf{v}_{3,1}^{(1)}\mathbf{w}_{3,1})$ can be cancelled by the cached content of user $1$, and Rank$(\mathbf{v}_{1,1}^{(1)},\mathbf{v}_{1,2}^{(1)})=2$. So user $1$ can decode the required packets $\mathbf{w}_{1,2}$ and $\mathbf{w}_{1,3}$. Similarly, users $2, 3$ can also decode their requested packets respectively. We can see that the server transmits $6$ packets to the users in each block. So the sum-DoF of the caching scheme is $6=2+4=\frac{GKZ}{F}+G\lceil\frac{L}{G}\rceil$.
\end{itemize}
\hfill $\square$
\end{example}

It is worth noting that the $(G,L,K,M,N)$ TST MIMO coded caching scheme with $G\leq {t+ \lceil L/G\rceil-1\choose t}$ in \cite{MMA} can be  represented by the following  $(G$, $L$, $K$, ${K\choose t}{K-t-1\choose  \lceil L/G\rceil-1}G$, ${K-1\choose t-1}{K-t-1\choose  \lceil L/G\rceil-1}G$, ${K\choose t+ \lceil L/G\rceil}{t+ \lceil L/G\rceil-1\choose t})$ TST MIMO-PDA 
\begin{align*}
 \mathbf{A}=(\mathbf{A}((\mathcal{T},\mathcal{R},l), k))_{\mathcal{T}\in {[K]\choose t},\mathcal{R}\in {[K-t-1]\choose  \lceil L/G\rceil-1},l\in[G],k\in [K]}
\end{align*} where for any $l\in[G]$, $\mathcal{R}\in {[K-t-1]\choose  \lceil L/G\rceil-1}$, $\mathcal{T}\in {[K]\choose t}$, and $k\in [K]$, the entry
\begin{align}
\label{eq-MN-array}
\mathbf{A}((\mathcal{T},\mathcal{R},l), k)=\left\{
\begin{array}{cc}
* & \ \ \ \text{if}  \ \ k\in \mathcal{T}\\
(\mathcal{S},\lceil\frac{o(\mathcal{S})}{G}\rceil) & \text{otherwise}
\end{array}
\right.
\end{align}
Here $\mathcal{S}=\mathcal{T}\cup ([K]\setminus(\{k\}\cup \mathcal{T}))|_{\mathcal{R}}\cup \{k\}$ and $o(\mathcal{S})$ is the occurrence order of $\mathcal{S}$ in the column $k$ from up to down. Let us take the following example to illustrate the above construction. In this paper, we represent each subset as a string for short. For instance, the set $\{1,2,3,4\}$ is written as $1234$. In addition, all the vectors are arranged in lexicographical order.
\begin{example}
\label{ex-2-3-4-2-MIMO}
When $G=2$, $L=3$, $K=4$ and $t=2$, we have the row labels
\begin{align*}
\mathcal{F}=\{&(12,34,1),(13,23,1),(14,23,1),(23,14,1),(24,12,1),(34,12,1),\\
&(12,34,2),(13,24,2),(14,23,2),(23,14,2),(24,13,2),(34,12,2)\}
\end{align*} and columns $[4]$. From \eqref{eq-MN-array} we have the following $(2,3,4,12,6,3)$ TST MIMO-PDA.
\begin{align}
	\label{eq-MIMO-PDA-K=4}
	\mathbf{A} = 
	\begin{blockarray}{ccccc}
		1 & 2 & 3 & 4 \\ 
		\begin{block}{(cccc)c} 
			*          &*        &1234,1       & 1234,1 & 12,34,1 \\
			*          &1234,1    &*           &1234,1 & 13,24,1 \\
			*          &1234,1    &1234,1     &* & 14,23,1 \\
			1234,1     &*         &*         &1234,2  & 23,14,1 \\
			1234,1     &*         &1234,2        & *    & 24,13,1 \\
			1234,2     &1234,2    &*           & * & 34,12,1 \\
			*         &*          &1234,2   &1234,2 & 12,34,2 \\
			*         &1234,2    &*         &1234,3 & 13,24,2 \\
			*         &1234,3    &1234,3    & * & 14,23,2 \\
			1234,2     &*        &*        &1234,3& 23,14,2 \\
			1234,3    &*         &1234,3     & * & 24,13,2 \\
			1234,3    &1234,3     &*     & * & 34,12,2 \\
		\end{block}
	\end{blockarray}
\end{align}  
By Theorem \ref{th-Fundamental}, we have a $(G=2,L=3,K=4,M,N)$ coded caching scheme with the memory ratio $\gamma=\frac{1}{2}$, subpacketization $F=12$ and sum DoF $g=8=G(t+\lceil\frac{L}{G}\rceil)=2(2+\lceil\frac{3}{2}\rceil)$.
\end{example}

\begin{remark}[The MN PDA]
	\label{remark-TST-PDA}When $G=L=1$, we have $\lfloor L/G\rfloor=0$ which implies that ${K-t-1\choose \lceil L/G\rceil-1}G=1$ and ${K\choose t+ \lceil L/G\rceil}{t+ \lceil L/G\rceil\choose t}={K\choose t+1}$. Then the TST MIMO-PDA is exactly the well known  $(K, {K\choose t}, {K-1\choose t-1}, {K\choose t+1})$ MN PDA \cite{MN,YCTC}.
\end{remark}
%\begin{example}
%	\label{example-PDA} When $G=L=1$, $K=4$ and $t=2$, we have the row labels
%	$\mathcal{F}=\{12,13,14,23,24,34\}$ and columns $[4]$. By remark \ref{remark-TST-PDA}, we have the following $(4,6,3,4)$ MN PDA.
%	\begin{align}
%		\mathbf{A} = 
%		\begin{blockarray}{ccccc}
%			1 & 2 & 3 & 4 \\ 
%			\begin{block}{(cccc)c} 
%				*          &*        &123       & 124 & 12 \\
%				*          &123    &*           &134 & 13 \\
%				*          &124    &134     &* & 14 \\
%				123     &*         &*         &234  & 23 \\
%				124     &*         &234        & *    & 24 \\
%				134     &234    &*           & * & 34 \\
%			\end{block}
%		\end{blockarray}
%	\end{align}
%\end{example}
By Theorem \ref{th-Fundamental}, the MIMO coded caching scheme can be obtained by constructing an appropriate MIMO-PDA, and the sum-DoF can be studied by counting the number of times each integer occurs in a MIMO-PDA. By counting the frequency of the used stars in a MIMO-PDA,  we can derive the maximum sum-DoF, the proof of which is included in Appendix \ref{proof-DoF}.

\begin{theorem}\label{th-DoF}
	The sum-DoF of a $(G,L,K,M,N)$ MIMO coded caching scheme with memory ratio $\gamma=\frac{Z}{F}$ realized by  any $(G,L,K,F,Z,S)$ MIMO-PDA  is no more than $\frac{GKZ}{F}+G\lceil L/G\rceil$.
	\hfill $\square$
\end{theorem}

%By the proof of Theorem \ref{th-DoF} in Appendix \ref{proof-DoF}, we have the following important investigation.
%\begin{proposition}
%	\label{pro-dof}
%	Given an optimal $(G,L,K,F,Z,S)$ MIMO-PDA $\mathbf{P}$, each row of $\mathbf{P}$ has exactly $t=KZ/F$ stars. In addition, for any positive integer $s\in[S]$ the subarray $\mathbf{P}^{(s)}$
%	including the rows and columns containing $s$ has the following property:
%	\begin{itemize}
%		\item Integer $s$ occurs exactly $G$ times in each column of $\mathbf{P}^{(s)}$;
%		\item In each row of $\mathbf{P}^{(s)}$, there are exactly $\tau=\lceil L/G\rceil$ integer entries while the other elements in this row is $*$. This implies that $\mathbf{P}^{(s)}$ has exactly $t+\tau$ columns.
%	\end{itemize}
%\end{proposition}

 In fact, the authors in \cite{MMA} also obtained the following result on the bound of sum-DoF.
\begin{lemma}[\cite{MMA}]
	\label{le-existing-DoF}
	Under the constraints of uncoded cache placement and one-shot linear delivery, the sum-DoF of an $(G,L,K,M,N)$ MIMO coded caching scheme is given by solving the following optimization problem:
	\begin{align}
		\label{eq-optimal}
		&\text{sum-DoF}_{\max}(\beta^*, \Omega^*) = \max_{\beta, \Omega} \Omega \beta,\nonumber\\
		\text{s.t.}\ \ \ \  &\beta \leq \min \left( G, \frac{L{\Omega-1\choose t}}{1+(\Omega-t-1){\Omega-1\choose t}} \right),
	\end{align}where $\Omega\leq K$ and $\beta\leq G$ represent  the number of the served users and the number of packets decoded by each user at each block respectively, and $\beta^*$ and $\Omega^*$ represent the optimal parameters chosen to achieve $\text{sum-DoF}_{\max}$.
	\hfill $\square$
\end{lemma}
It is worth noting that when $\beta=G$ and $G\leq {t+ \lceil L/G\rceil-1\choose t}$, the maximum sum-DoF is $\frac{GKZ}{F}+G\tau$ in \eqref{eq-optimal} which is the exactly value in Theorem \ref{th-DoF}. In addition, we can check that the TST MIMO-PDA has the maximum sum-DoF by Theorem \ref{th-DoF}. This fact has also been showed in \cite{TSJA}.

For the convenience, in this paper a MIMO-PDA is said to be optimal if it achieves the maximum sum-DoF derived in Theorem \ref{th-DoF}. In addition, a scheme generalized by a MIMO-PDA is referred to as a MIMO-PDA scheme, and a MIMO-PDA scheme is considered optimal if it achieves the maximum sum-DoF derived in Theorem \ref{th-DoF}. Clearly, the TST scheme is optimal. In the following, we will show that the constraint of TST MIMO-PDA constructed in \cite{TSJA} can be further relaxed. First, the following concept of a MIMO-PDA is useful.
\begin{definition}[Consistency number]\rm
\label{def-consisten}
Given a MIMO-PDA $\mathbf{P}=(\mathbf{P}(f,k))_{f\in[F],k\in[K]}$ containing $S$ different integers, the \emph{consistency number} of $\mathbf{P}$, denoted $\mu$, is the largest nonnegative integer such that for every $s\in[S]$ and every column $k\in[K]$, there exist at most $\mu$ distinct rows $f_{i_1},\dots,f_{i_\mu}$ of $\mathbf{P}^{(s)}$ satisfying $\mathbf{P}^{(s)}(f_{i_1},k)= \cdots =\mathbf{P}^{(s)}(f_{i_\mu},k)=s$ and $\mathcal{P}^{(s)}_{f_{i_1}} = \cdots = \mathcal{P}^{(s)}_{f_{i_\mu}}$. 
\end{definition} 

By Definition \ref{def-consisten}, we can see that the consistency number of a MIMO-PDA with parameters $G$ and $L$ is less than or equal to $\langle L\rangle_G$. 
We now consider the consistency number of a TST MIMO-PDA $\mathbf{A}$. We know that each subset $\mathcal{S}$ in \eqref{eq-MN-array} contains exactly  $t+\lceil L/G\rceil$ elements, and the subarray $\mathbf{A}^{(\mathcal{S})}$, consisting of rows and columns that include $\mathcal{S}$, has ${t+\lceil L/G\rceil\choose t}G$ rows and $t+\lceil L/G\rceil$ columns. In each column of $\mathbf{A}^{(\mathcal{S})}$, there are exactly $\binom{K-t-1}{\lfloor L/G\rfloor}G$ non-star entries, and each such entry contains the subset $\mathcal{S}$. Moreover, fix a column of $\mathbf{A}^{(\mathcal{S})}$. For every row corresponding to one of the $\binom{K-t-1}{\lfloor L/G\rfloor}G$ non-star entries in this column, consider the set of column indices in $\mathbf{A}^{(\mathcal{S})}$ where a non-star entry appears in that same row. This yields $\binom{K-t-1}{\lfloor L/G\rfloor}G$ such column-index sets. Among these, exactly $\binom{K-t-1}{\lfloor L/G\rfloor}$ are distinct. Based on the second coordinate $\lceil\frac{o(\mathcal{S})}{G}\rceil$ of an non-star entry defined in \eqref{eq-MN-array}, the consistency number can be derived as follows. 
\begin{proposition}
\label{pro-rho}
The consistency number of a TST MIMO PDA with parameters $K$, $t$, $G$ and $L$ is $\mu =\lceil G/{t+\lfloor L/G\rfloor\choose t}\rceil$.
\hfill $\square$ 
\end{proposition}Clearly, $\lceil G/\gamma\rceil\leq G$ always holds. Hence, the condition for TST MIMO-PDA $G \leq \binom{t + \lfloor L/G \rfloor}{t}$ in \cite{TSJA} can be relaxed to $\lceil G/{t+\lfloor L/G\rfloor\choose t}\rceil \leq\langle L\rangle_G$. Then we can relax the constraint of the TST MIMO-PDA constructed in \cite{TSJA} as follows. 
\begin{theorem}[\cite{TSJA}]\rm
\label{th-TST}
For any positive integers $K$, $t$, $G$ and $L$, if $\mu=\lceil G/{t+\lfloor L/G\rfloor\choose t}\rceil \leq\langle L\rangle_G$, there exists an optimal $(G$, $L$, $K$, ${K\choose t}{K-t-1\choose  \lceil L/G\rceil-1}G$, ${K-1\choose t-1}{K-t-1\choose  \lceil L/G\rceil-1}G$, ${K\choose t+ \lceil L/G\rceil}{t+ \lceil L/G\rceil-1\choose t})$ TST MIMO-PDA  with consistency number $\mu$.
\end{theorem}

Finally, we should point out that subpacketization of a TST MIMO-PDA increases exponentially with the number of users. In the rest of this paper, we will construct three classes of new optimal MIMO-PDAs with lower subpacketization.

\section{Optimal MIMO-PDA Schemes with Lower Subpacketization}
\label{sec-const}
In this section, we first propose a cyclic construction with linear subpacketization under a strong parameter constraint. To relax this constrain, we introduce a hybrid construction with lower subpacketization. Some theoretical and numerical comparisons on the subpacketization are obtained.   
\subsection{Optimal MIMO-PDAs}
\label{sub-three-optimal}
By Definition \ref{def-MIMO-PDA}, the key to constructing the MIMO-PDA lies in Condition C$4$. That is, for any integer $s \in [S]$, each row of the subarray $\mathbf{P}^{(s)}$ contains at most $\tau=\lceil\frac{L}{G}\rceil$ integer entries. Furthermore, no more than $\rho$ rows in $\mathbf{P}^{(s)}$ share the same set of integer entry position set. Based on this investigation, we obtain our main results. 

Using the star-cyclic placement, we obtain the following result, whose proof is included in Appendix \ref{app-proof-latin}. 
\begin{theorem}[Square MIMO-PDA]\label{th-latin}
 For any positive integers $K$, $G$, $L$, and $t$ satisfying $K\leq \tau+t$ and $\langle L\rangle_G\leq K-t$ where $\tau=\lceil L/G\rceil$, there exists an optimal $(G,L,K,F=GK,Z=Gt,S=K-t)$ MIMO-PDA, which realizes a $(G,L,K,M,N)$ MIMO coded caching scheme with the memory ratio $\gamma=t/K$ and the subpacketization $F=G K$.
\hfill $\square$
\end{theorem} 

The subpacketization level of the scheme presented in Theorem~\ref{th-latin} scales linearly with the number of users. However, this construction is subject to a stringent limitation, namely that $K \leq \tau + t$. In the following, we focus on studying the case $K> \tau + t$. We introduce a grouping method anda matching in a bipartite graph to construct new optimal MIMO-PDAs. Then we can obtain the following result whose proof is included in Appendix \ref{proof-grouping}. 
\begin{theorem}[Grouping MIMO-PDA]\label{th-grouping}
	Given an optimal $(G,L_1,K_1,F,Z,S)$ MIMO-PDA with the consistency number $\mu$, there exists an optimal $(G,L,K=mK_1,F,Z,S)$ MIMO-PDA where $m\lceil\frac{L_1}{G}\rceil=\lceil\frac{L}{G}\rceil$ for any integer $L\geq L_1$ satisfying $\langle L\rangle_G\geq \mu$.\hfill $\square$ 
\end{theorem}

For instance, by Theorem \ref{th-TST} and Theorem \ref{th-grouping} the following new optimal schemes can be obtained.
\begin{corollary}[GTST MIMO-PDA]
	\label{cor-TST-grouping}For any positive integers $K_1$, $t_1$, $G$, $m$, $L_1$ and $L$ satisfying that  $\mu=\lceil G/{t_1+ \lceil L/G\rceil-1\choose t_1}\rceil \leq\langle L\rangle_G$ and $m\lceil \frac{L_1}{G}\rceil=\lceil\frac{L}{G}\rceil$, there exists an optimal $(G$, $L$, $K=mK_1$, $F={K_1\choose t_1}{K_1-t_1-1\choose \lceil L/G\rceil-1}G$, $Z={K_1-1\choose t_1-1}{K_1-t_1-1\choose \lceil L/G\rceil-1}G$, $S={K_1\choose t_1+ \lceil L/G\rceil}{t_1+ \lceil L/G\rceil-1\choose t_1})$ MIMO-PDA which generalizes a $F$-division $(G,L,K,M,N)$ MIMO coded caching scheme with the memory ratio $M/N=t_1/K_1$ and sum-DoF $\min\{Gt+G\lceil \frac{L}{G}\rceil,GK\}$.  
	\hfill $\square$  
\end{corollary}

In the following subsection, we demonstrate that the MIMO-PDA in Corollary \ref{cor-TST-grouping} achieves a significant reduction in subpacketization compared to the TST MIMO-PDA under identical parameters. However, Corollary \ref{cor-TST-grouping} imposes a strong constraint: $m\lceil\tfrac{L_1}{G}\rceil = \lceil\tfrac{L}{G}\rceil$. This implies that a new optimal MIMO-PDA with parameter $L$ can be constructed only if an optimal TST MIMO-PDA with parameter $L_1$ is available such that $\lceil\tfrac{L}{G}\rceil$ is divisible by $\lceil\tfrac{L_1}{G}\rceil$. Building upon an extended grouping method, and by utilizing the uniform complete hypergraph factorization Lemma \cite{BZ} in conjunction with the perfect matching Lemma from \cite{BJM}, we can relax this constraint.

The main construction of the MIMO-PDA is that we first construct two MIMO-PDAs $\mathbf{X}$ and $\mathbf{Y}$ with $K$ users such that  $\mathbf{X}$ and $\mathbf{Y}$ have sum-DoF $g_1=mG(t_1+\lceil L_1/G\rceil)$ and 
$g_2=G(\lceil L/G\rceil-m\lceil L_1/G\rceil)$ respectively. We obtain the desired matrix $\mathbf{P}$ by first vertically concatenating $\mathbf{X}$ and $\mathbf{Y}$, and then adjusting the integers in $\mathbf{X}$ and $\mathbf{Y}$ using the perfect matching Lemma in \cite{BJM}. Clearly the obtained MIMO-PDA $\mathbf{P}$ has the sum-DoF $g=g_1+g_2=G(mt_1+\lceil L/G\rceil)$ which achieves the upper bound on the sum-DoF. During the adjustment of the integer entries, it must be ensured that rows containing the same integer in $\mathbf{X}$ and $\mathbf{Y}$ respectively have different star positions. This is the key point to compute the consistency number which is ensure the Condition C4-b) of MIMO-PDA. In addition, in order to make the subpacketization as small as possible in our constructing method, we use the hypergraph factorization Lemma \cite{BZ} is used to construct the MIMO-PDAs $\mathbf{X}$ and $\mathbf{Y}$. That is the following result, whose proof is provided in Section \ref{proof-regular}.
\begin{theorem}[Hybrid MIMO-PDA]\label{th-MAPDA-hypergraph}	
For any positive integers $G$, $L_1$, and $L$, let $\tau_1=\lceil L_1/G\rceil$, $\tau=\lceil L/G\rceil$, $m=\lfloor \frac{\tau}{\tau_1}\rfloor\geq 1$, and $\tau_2=\tau-m\tau_1$. When $\tau_2>0$, for any positive integers $K_1$ and $t_1$ satisfying $\tau_1|\{t_1$, $K_1\}$, $t_1+\tau_1<K_1$, $\tau_2|(t_1+\tau_1)$, and $2G\leq {t_1+\tau_1-1\choose t_1 }$, there exists an optimal $(G$, $L$, $mK_1$, $F$, $Z$, $S)$ MIMO-PDA with
\begin{align}
F&=\frac{G\Lambda_{1}\Lambda_{3}}{\gcd(G,\Lambda_{2})}\left(\frac{m(t_1+\tau_1)}{\tau_2}+1 \right){K_1\choose t_1},\nonumber\\
Z&=\frac{G\Lambda_{1}\Lambda_{3}}{\gcd(G,\Lambda_{2})}\left(\frac{m(t_1+\tau_1)}{\tau_2}+1 \right){K_1-1\choose t_1-1},\nonumber\\ 
S&=\frac{m\Lambda_{2}\Lambda_{3}}{\gcd(G,\Lambda_{2})}\frac{t_1+\tau_1}{\tau_2}{K_1\choose t_1+ \tau_1}
\label{eq-parameters}
\end{align}where $\Lambda_{1}={K_1-t_1-1\choose  \tau_1-1}$,  $\Lambda_{2}={t_1+\tau_1-1\choose  \tau_1-1}$, and $\Lambda_{3}={t_1+\tau_1-1\choose  \tau_2-1}$.
\hfill $\square$
\end{theorem}
\begin{remark}[Explanation of the case $G=1$]
\label{remark-special-case-1}
The TST MIMO-PDA construction achieves a consistency number of $1$ only under the condition $G \leq \binom{t_1+\tau_1-1}{t_1}$. In contrast, our hybrid construction is not expressible in the simple form $\mu = \big\lceil G / \binom{t + \lfloor L/G\rfloor}{t} \big\rceil$ used in Theorem~\ref{th-TST}. To obtain an MIMO-PDA applicable for arbitrary $L$, we provide in Theorem~\ref{th-MAPDA-hypergraph} a design with $\mu = 1$. In this setting, the construction requires the constraints $t_1+\tau_1 < K_1$ and $2G \leq \binom{t_1+\tau_1-1}{t_1}$, together with the divisibility conditions $\tau_1 \mid {t_1, K_1}$ and $\tau_2 \mid (t_1+\tau_1)$.

When $G = 1$, the MIMO-PDA reduces to an MA-PDA, for which the consistency number is irrelevant. Applying the same proof technique yields an optimal MA-PDA that retains only $t_1+\tau_1 < K_1$, thereby eliminating all other previous restrictions: $\tau_1 \mid {t_1, K_1}$, $\tau_2 \mid (t_1+\tau_1)$, and $2G \leq \binom{t_1+\tau_1-1}{t_1}$. In addition, when $m=\tau/\tau_1$ in Theorem \ref{th-MAPDA-hypergraph} the obtained optimal MIMO-PDA is exactly the optimal MAPDA in \cite{LEEP}; when $m<\tau/\tau_1$, the obtain optimal MIMO-PDA is exactly the optimal MAPDA in \cite{YWCC}.\hfill $\square$
\end{remark}

\begin{remark}[The extension of the limitation $\tau_2 \mid (t_1+\tau_1)$]
Recall that we obtain the MIMO-PDA in Theorem~\ref{th-MAPDA-hypergraph} by combining two MIMO-PDAs: $\mathbf{X}$ with sum-DoF $g_1$ and $\mathbf{Y}$ with sum-DoF $g_2$. Each element in $\mathbf{Y}$ is obtained by splitting the columns occupied by the same element in a base TST MIMO-PDA across $t_1+\tau_1$ columns into $g_2$ parts, which originally required the condition $\tau_2 \mid (t_1+\tau_1)$.

Since $\mathbf{Y}$ can be viewed as consisting of $m$ horizontally replicated TST MIMO-PDAs, it is possible to split the $g_2$ parts from columns corresponding to identical vectors across these $m$ replicas. This perspective relaxes the original requirement $\tau_2 \mid (t_1+\tau_1)$. However, designing such a splitting while ensuring the final MIMO-PDA $\mathbf{P}$ remains optimal—particularly in applying the perfect matching lemma from \cite{BJM}—introduces significant technical challenges.\hfill $\square$
\end{remark}
\subsection{Performance Analyses for Subpacketization} 
\label{subsec-performance}
The main contribution of Theorem \ref{th-MAPDA-hypergraph} is to reduce the subpacketzation for
the case $\tau+t<K$ and $G\nmid L$ where $\tau=\lceil L/G\rceil$, while achieving the maximum sum-DoF. In the following, we compare the performance of the hybrid scheme in Theorem \ref{th-MAPDA-hypergraph} with the TST scheme in Theorem~\ref{th-TST}. Theoretical and numerical comparisons show that the hybrid scheme in Theorem~\ref{th-MAPDA-hypergraph} achieves a remarkable exponential reduction in subpacketization compared to the TST scheme in Theorem~\ref{th-TST}.  

Let $\tau = \lceil L/G \rceil$. Theorem~\ref{th-TST} provides a TST-MIMO PDA with subpacketization $F_{\text{TST}} = G \binom{K}{t} \binom{K-t-1}{\tau-1}$ and cache ratio $M_{\text{TST}}/N = t/K$, where $t$ is a positive integer satisfying $t < K_1 - \tau$. Theorem~\ref{th-MAPDA-hypergraph} provides a hybrid MIMO-PDA with the subpacketization $F_{\text{Hybrid}}= \frac{G\Lambda_{1}\Lambda_{3}}{\gcd(G,\Lambda_{2})}(\frac{m(t_1+\tau_1)}{\tau_2}+1){K_1\choose t_1}$ and cache ratio $\frac{M_{\text{Hybrid}}}{N}=\frac{t_1}{K_1}$, where $\tau_1 = \lceil L_1/G \rceil$, $m = \lfloor \tau/\tau_1 \rfloor$, $\tau_2 = \tau - m\tau_1$,  $\Lambda_{1}={K_1-t_1-1\choose  \tau_1-1}$,  $\Lambda_{2}={t_1+\tau_1-1\choose  \tau_1-1}$, and $\Lambda_{3}={t_1+\tau_1-1\choose  \tau_2-1}$. To match the total user count $K$ and the cache ratio, we set $K = mK_1$ and $t = m t_1$, where $t_1+\tau_1 <K_1$. Then, the subpacketization ratio can be written as follows. 
\begin{align}
\frac{F_{\text{Hybrid}}}{F_{\text{TST}}}=&\frac{\frac{G\Lambda_{1}\Lambda_{3}}{\gcd(G,\Lambda_{2})}(\frac{m(t_1+\tau_1)}{\tau_2}+1){K_1\choose t_1}}{G \binom{K}{t} \binom{K-t-1}{\tau-1}}=
 \frac{m \binom{t_1+\tau_1}{\tau_2} + \binom{t_1+\tau_1-1}{\tau_2-1}}{\gcd(G,\Lambda_{2})} \cdot \frac{\binom{K_1}{t_1}}{\binom{mK_1}{mt_1}} \cdot \frac{\binom{K_1-t_1-1}{\tau_1-1}}{\binom{mK_1-mt_1-1}{\tau-1}}\nonumber\\
\approx& \frac{m \binom{t_1+\tau_1}{\tau_2} + \binom{t_1+\tau_1-1}{\tau_2-1}}{\gcd(G,\Lambda_{2})} \cdot \frac{\sqrt{m}}{2^{(m-1)K_1 H(\gamma)}} \cdot \frac{\binom{K_1-t_1-1}{\tau_1-1}}{\binom{K-t-1}{\tau-1}}\ \ \ \ \ \ \ \ \ (K_1\rightarrow \infty)\nonumber\\
\approx &\frac{m \binom{t_1+\tau_1}{\tau_2} + \binom{t_1+\tau_1-1}{\tau_2-1}}{\gcd(G,\Lambda_{2})} \cdot \frac{\sqrt{m}}{2^{(m-1)K_1 H(\gamma)}} \cdot\frac{(K_1(1-\gamma))^{\tau_1-1}/(\tau_1-1)!}{(mK_1(1-\gamma))^{\tau-1}/(\tau-1)!} \nonumber \\ 
=&\frac{m \binom{t_1+\tau_1}{\tau_2} + \binom{t_1+\tau_1-1}{\tau_2-1}}{\gcd(G,\Lambda_{2})} \cdot \frac{\sqrt{m}}{2^{(m-1)K_1 H(\gamma)}} \cdot \frac{(\tau-1)!}{(\tau_1-1)!} \cdot m^{-(\tau-1)} \cdot (K_1(1-\gamma))^{-(\tau-\tau_1)}\nonumber\\
\approx&O\left( 2^{-(m-1)K_1 H(\gamma)} \cdot K_1^{-((m-1)\tau_1 + \tau_2)} \right).	\label{eq:R_ratio}
\end{align}Here $H(\gamma) = -\gamma\log_2\gamma - (1-\gamma)\log_2(1-\gamma)$. The exponential term $2^{-(m-1)K_1 H(\gamma)}$ dominates for large $K_1$, indicating that the hybrid scheme achieves an exponential reduction in subpacketization compared to the TST scheme. The reduction becomes more significant with larger grouping parameter $m$ and larger per-group user count $K_1$. 

The superior performance of the hybrid scheme stems from its hierarchical structure. First, dividing $K$ users into $m$ groups of $K_1$ users reduces the combinatorial complexity from $\binom{K}{t}$ to $\binom{K_1}{t_1}$. Second, choosing $\tau_1 < \tau$ reduces the exponent in the second binomial factor, further lowering subpacketization. Third, while reducing subpacketization, the hybrid scheme maintains the same DoF performance as the TST scheme under the parameter constraints in Theorem~\ref{th-MAPDA-hypergraph}.

Finally,  we present numerical comparisons of subpacketization levels in Fig.~\ref{fig:performance} between the hybrid MIMO-PDA in Theorem \ref{th-MAPDA-hypergraph} and the TST MIMO-PDA in Theorem \ref{th-TST}. In Fig.~\ref{fig:performance}, the blue solid curve with circle markers represents $F_{\text{TST}}$ from Theorem~\ref{th-TST}, while the red dashed curve with square markers represents $F_{\text{Hybrid}}$ from Theorem~\ref{th-MAPDA-hypergraph}. A substantial vertical gap between the two curves is observed, quantitatively demonstrating the advantage of Theorem~\ref{th-MAPDA-hypergraph} in reducing subpacketization. For instance, when $K=120$, we obtain $F_{\text{TST}} \approx 2.1\times 10^{24}$ and $F_{\text{Hybrid}} \approx 3.8\times 10^{12}$, giving $F_{\text{Hybrid}}/F_{\text{TST}} \approx 1.8\times 10^{-12}$—a reduction by twelve orders of magnitude.  
\begin{figure}[h]
	\centering
	\includegraphics[width=0.8\textwidth]{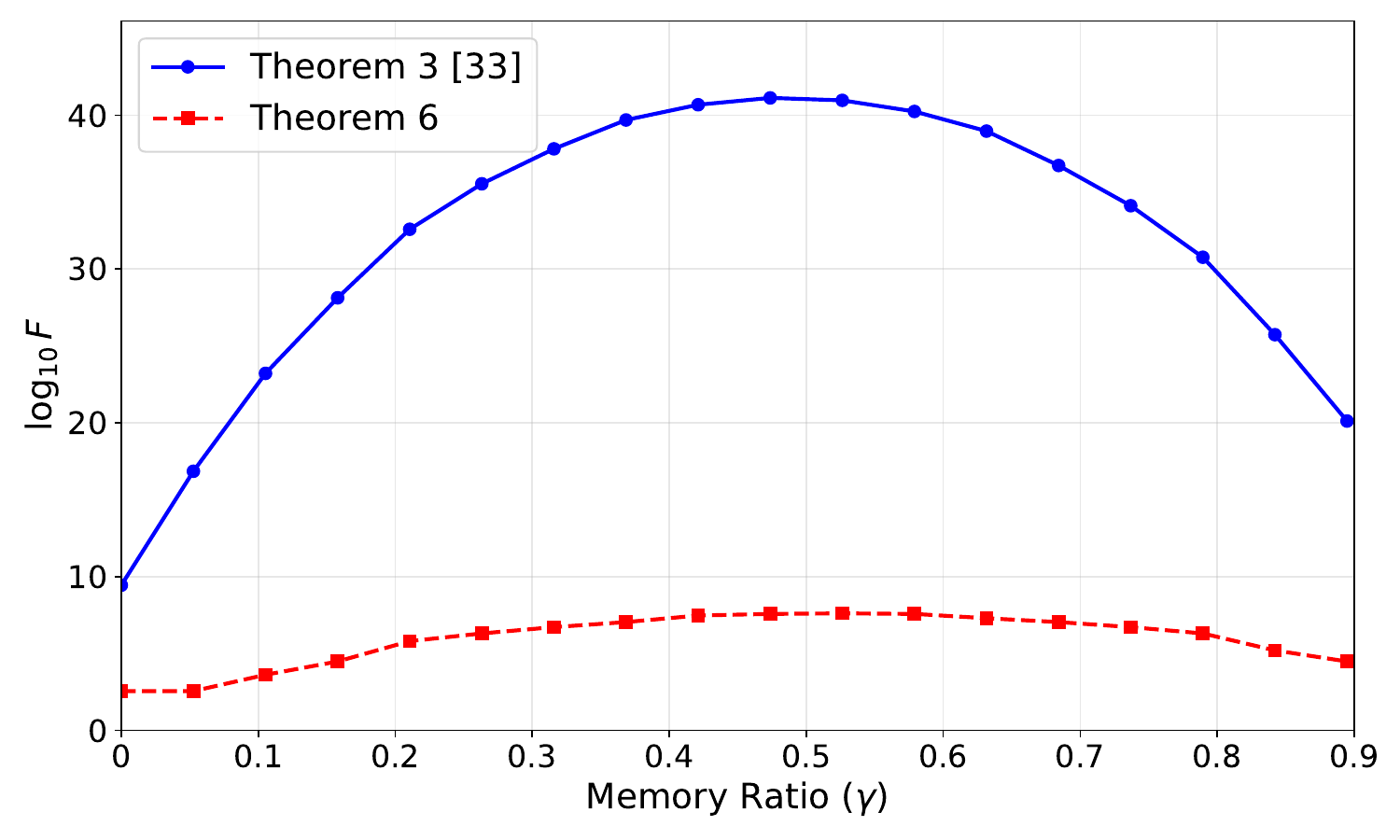}
	\caption{Subpacketization Comparison Between the MIMO-PDAs in Theorems \ref{th-TST} and \ref{th-MAPDA-hypergraph} with $K=120$, $G=3$, and $L=13$.}
	\label{fig:performance}
\end{figure}  

\subsection{Sketch of the Proposed MIMO-PDA in Theorem \ref{th-latin}}
For any positive integers $t$, $K$ and $\tau$ satisfying $t+\tau\geq K$, it is easy to check that a square $\mathbf{B}$ with size $K$ composed by symbols $*$ and null always satisfies the condition C$4$-a) of Definition \ref{def-MIMO-PDA} if its each row has exactly $t$ stars even when it contains the same integer. In addition, if the stars in each row are placed in a cyclic wrap-around topology\footnote{For each integer $i\in[K]$, the stars in the $i$-th row are placed $t$ cyclically from the $i$-th to the $\langle i+t\rangle_{K}$-th position.}, each row has a unique integer position set, i.e., the Condition C$4$-b) of Definition \ref{def-MIMO-PDA} always holds. In this case, each column of $\mathbf{B}$ has exactly $K-t$ null entries.  When $G\leq K-t$, we can replicate $\mathbf{B}$ vertically $G$ times and sequentially fill every $G$ consecutive null entries in each column with the same integer. Clearly each integer occurs in each column exactly $G$ times, i.e, that the Condition C$3$ holds. Let us take the following example to further illustrate our construction idea.

When $G=2$, $L=3$, $K=4$ and $t=2$, as listed in Fig. \ref{fig-flow-1} we first obtain a square with size of $K=4$ and place the stars in each row under a cyclic wrap-around topology to obtain the array $\mathbf{B}$. Then replicate  $\mathbf{B}$ and fill the first two null entries in each column with integer $1$ and the left two null entries in each column with integer $2$ respectively, we have our desired array. We can check that it is a $(2,3,4,8,4,2)$ MIMO-PDA.
\begin{figure}[http!]
\centering
\includegraphics[width=0.7\linewidth]{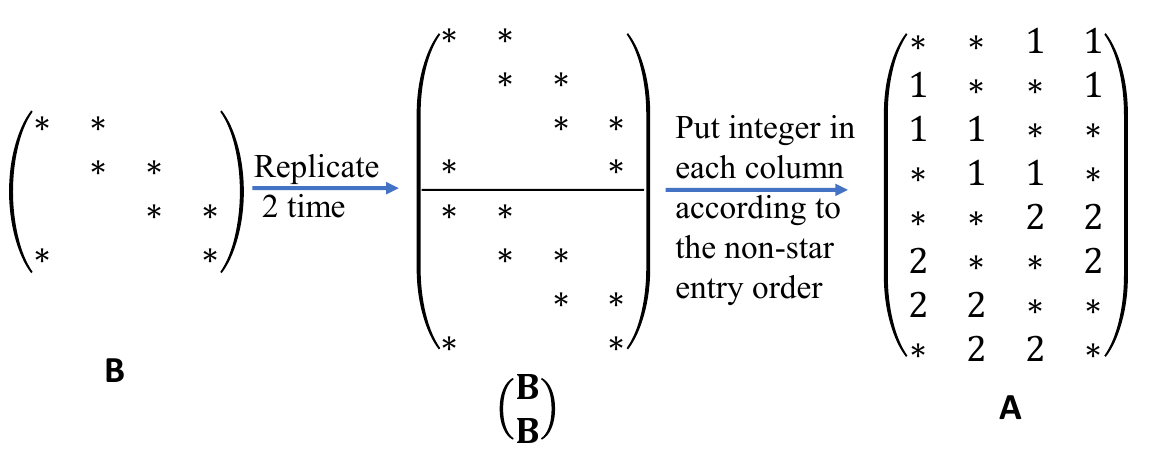}
\caption{The flow of constructing square $(2,3,4,8,4,2)$ MIMO-PDA.}
\label{fig-flow-1}
\end{figure}
By Theorem \ref{th-Fundamental}, we have a $(G=2,L=3,K=4,M,N)$ coded caching scheme with the memory ratio $\gamma=\frac{1}{2}$, subpacketization $F=8$ and the optimal sum DoF $g=8=G(t+\lceil\frac{L}{G}\rceil)=2(2+\lceil\frac{3}{2}\rceil)$.

\subsection{Sketch of MIMO-PDA in Theorem \ref{th-MAPDA-hypergraph}}
\label{subsection-example-general} 
The main idea of our hybrid construction is as follows. Using the Baranyai's Theorem \cite{Baranyai1975}, we first construct a new TST MIMO-PDA $\mathbf{B}$ by reputing a new integer in the second coordinate of each non-star entry of TST MIMO-PDA defined in \eqref{eq-MN-array}. By replicating $\mathbf{B}$ horizontally and vertically with appropriate integer adjustments, we construct two MIMO-PDAs, $\mathbf{X}$ and $\mathbf{Y}$, characterized by the same number of users $K$ and the same integer set $[S]$, but with different sum-DoF values: $g_1 = mG\left(t_1 + \tau_1 \right)$ and $g_2 = G\tau_2$, respectively. Finally, using the perfect matching \cite{BJM} we obtain the desired $g=g_1+g_2=G(mt_1+\tau)$-MIMO-PDA $\mathbf{P}$ by first vertically concatenating $\mathbf{X}$ and $\mathbf{Y}$. Hence, our construction consists of the following four steps. 

The following results and notations are useful in this paper. Let us first introduce a important result in the field of factorization of the hypergraph.
\begin{lemma}[Baranyai's Theorem \cite{Baranyai1975}]
	\label{lem-factor}
	For any set $\mathcal{V}$ of size $v$ and any positive integer $\alpha$ satisfying $\alpha|v$, all the subsets of ${\mathcal{V}\choose \alpha}$ can be partitioned  into $\Lambda={v-1\choose \alpha-1}$ parallels $\mathcal{F}_{\mathcal{V},1}$, $\mathcal{F}_{\mathcal{V},2}$, $\ldots$, $\mathcal{F}_{\mathcal{V},\Lambda}$ of $\mathcal{V}$.
\end{lemma} 

 A graph can be denoted by $\mathcal{G}=(\mathcal{V},\mathcal{E})$, where $\mathcal{V}$ is the set of vertices and $\mathcal{E}$ is the set of edges, and a subset of edges $\mathcal{M}\subseteq \mathcal{E}$ is a matching if no two edges have a common vertex. A bipartite graph, denoted by $\mathcal{G}=(\mathcal{X},\mathcal{Y}, \mathcal{E})$, is a graph whose vertices are divided into two disjoint parts $\mathcal{X}$ and $\mathcal{Y}$ such that every edge in $\mathcal{E}$ connects a vertex in $\mathcal{X}$ to one in $\mathcal{Y}$. For a set $\mathcal{X}'\subseteq \mathcal{X}$, let $N_{\mathcal{G}}(\mathcal{X}')$ denote the set of all vertices in $\mathcal{Y}$ adjacent to some vertices of $\mathcal{X}'$. The degree of a vertex is defined as the number of vertices adjacent to it. If every vertex of $\mathcal{X}$ has the same degree, we  call it the degree of $\mathcal{X}$ and denote it $d(\mathcal{X})$. A matching is called perfect if $|\mathcal{X}|=|\mathcal{Y}|$. 
\begin{lemma}[Perfect matching \cite{BJM}]
	\label{le-hall-theorem}
	Given a bipartite graph $\mathcal{G}=(\mathcal{X},\mathcal{Y},\mathcal{E})$ with $|\mathcal{X}|=|\mathcal{Y}|$, if there exist two positive integers $m$ and $n$ such that $d(\mathcal{X})=m$ and $d(\mathcal{Y})=n$, then there is a saturating matching for $\mathcal{X}$ if $m\geq n$, \emph{i.e.}, there exists a perfect matching.
\end{lemma}

In Fig. \ref{fig-flow-2}, by Lemmas \ref{lem-factor} and \ref{le-hall-theorem} we take an optimal $(G=2,L=13,K=24,F=7980,Z=3990,S=2520)$ MIMO-PDA as an example to introduce these four steps.  
\begin{figure}[http!]
	\centering
	\includegraphics[width=1\linewidth]{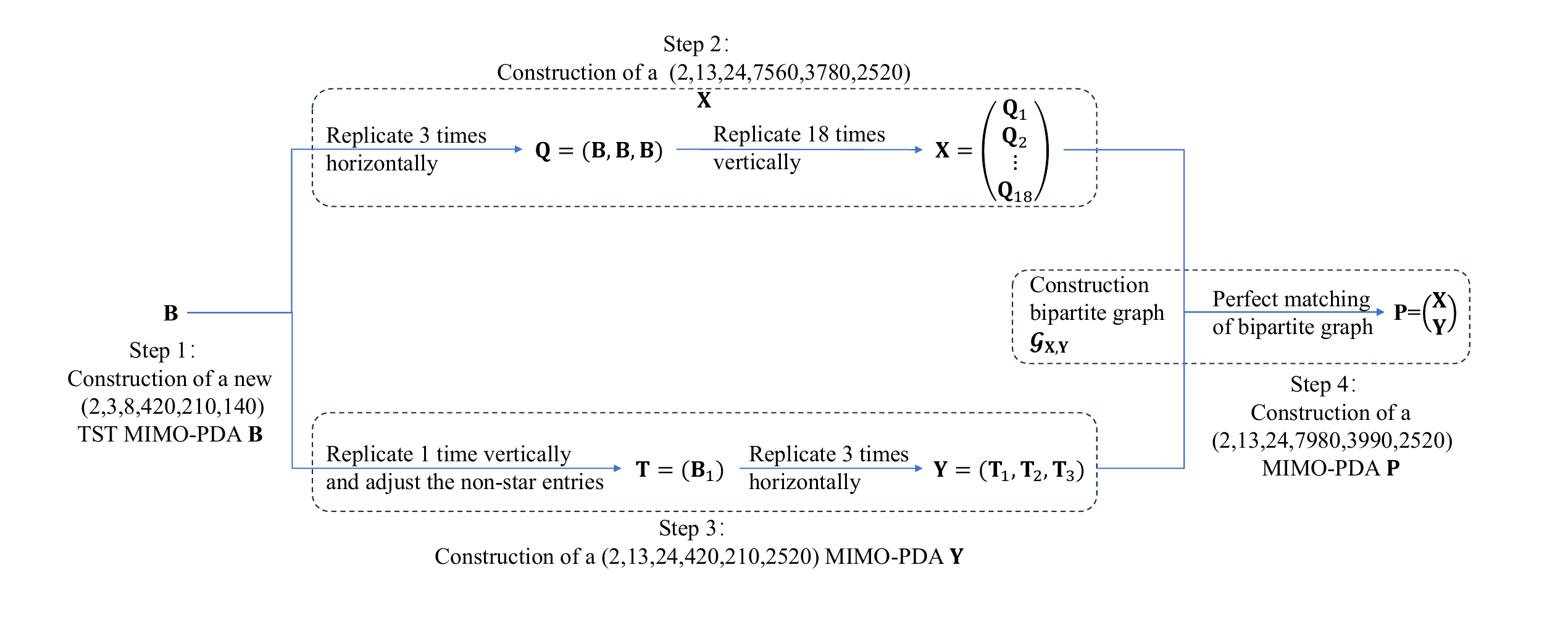}
	\caption{The ﬂow diagram of constructing $\mathbf{B}$, $\mathbf{Q}$, $\mathbf{X}$, $\mathbf{T}$, $\mathbf{Y}$ and $\mathbf{P}$ .}
	\label{fig-flow-2}
\end{figure} 

$\bullet$ {\bf Step 1} Construction of a new $(2,3,8,420,210,140)$ TST MIMO-PDA $\mathbf{B}$: By our hypothesis, we have $\tau=\lceil L/G\rceil=\lceil 13/2\rceil=7$. When $K_1=8$, $L_1=3$, $t_1=4$, we have $\tau_1=\lceil L_1/G\rceil=\lceil 3/2\rceil=2$ and $m=K/K_1=\lfloor \tau/\tau_{1}\rfloor=3$. In addition, $\tau_1|(t_1+\tau_1)$ since $t_1+\tau_1=4+2=6$. 

We now introduce the construction of $\mathbf{B}$. Its rows are indexed by $(\mathcal{T},r,l)$ where 
$\mathcal{T} \in \binom{[K_1]}{t_1} = \binom{[8]}{4}$, $r \in[\Lambda_{1}]=[\binom{K_1-t_1-1}{\tau_1-1}] =[3] $, $l \in \bigl[\frac{G}{\gcd(G,\Lambda_{2})}\bigr] = [2]$; its columns are indexed by $k \in [K_1]$. Clearly there are exactly 
$$F_1=\frac{G}{\gcd(G,\Lambda_{2})}\Lambda_{1}{K_1\choose t_1}=2\cdot{3\choose 1}\cdot{8\choose 4}=420$$
rows and $8$ columns in $\mathbf{B}$. Following the star‑entry convention in \eqref{eq-MN-array}, the entry at row $(\mathcal{T},r,l)$ and column $k$ is set to “$*$” if $k \in \mathcal{T}$. We can check that each column has exactly 
$$Z_1=\frac{G}{\gcd(G,\Lambda_{2})}\Lambda_{1}{K_1-1\choose t_1}=2\cdot{3\choose 1}\cdot{7\choose 4}=210 $$ 
stats. In the following, we only need to consider the putting vectors strategy for the non-star entries of $\mathbf{B}$, i.e., the case $k\not\in\mathcal{T}$. 

Let us consider the entries at the row $(\mathcal{T}=[4],r=1,l=1)$. By Lemma \ref{lem-factor}, we can divide all the $\tau_1=2$-subset of $\mathcal{R}=[K_1]\setminus\mathcal{T}=[8]\setminus[4]=[5:8]$ into $\Lambda_{1}={8-4-1\choose 2-1}=3$ parallels of $\mathcal{R}$ as follows. 
\begin{align}
	\mathcal{F}_{\mathcal{R},1}& = \{\{5, 6\}, \{7, 8\}\}\ \ 
	\mathcal{F}_{\mathcal{R},2} = \{\{5, 7\}, \{6, 8\}\}\ \
	\mathcal{F}_{\mathcal{R},3} = \{\{5, 8\}, \{6, 7\}\}.
	\label{eq-paralles-example-R}
\end{align} According to our hypothesis {\color{blue}$r=1$}, let us consider the first parallel $\mathcal{F}_{\mathcal{R},{\color{blue}1}}$. 
\begin{itemize}
\item[$-$] For the first element of $\mathcal{F}_{\mathcal{R},1}$, i.e., $\{5,6\}$, define  $\mathcal{S}=\mathcal{T}\cup \{5,6\}=[6]$. By Lemma \ref{lem-factor},  we can divide all the $\tau_1=2$-subset of $\mathcal{S}$ into $\Lambda_{2}={6-1\choose 2-1}=5$ parallels of $\mathcal{S}$, i.e.,  
\begin{align}
\mathcal{F}_{\mathcal{S},1}& = \{\{1, 2\}, \{3, 4\}, \{5, 6\}\}\ \ 
\mathcal{F}_{\mathcal{S},2}= \{\{1, 3\}, \{2, 5\}, \{4, 6\}\}\ \
\mathcal{F}_{\mathcal{S},3}= \{\{1, 4\}, \{2, 6\}, \{3, 5\}\}\nonumber\\
\mathcal{F}_{\mathcal{S},4}&=\{\{1, 5\}, \{2, 4\}, \{3, 6\}\}\ \
\mathcal{F}_{\mathcal{S},5}=\{\{1, 6\}, \{2, 3\}, \{4, 5\}\}.
\label{eq-paralles-example-S}
\end{align}Since $\{5,6\}$ is the third element of the first parallel, i.e., $\mathcal{F}_{\mathcal{S},{\color{red}1}}$, we put the vector $$\left(\mathcal{S},\left\lceil\frac{{\color{red}1}+(l-1)\Lambda_{2}}{G}\right\rceil\right)=\left([6],\left\lceil\frac{1+(1-1)\cdot 5}{2}\right\rceil\right)=([6],1)$$ into the entries at row $(\mathcal{T}=[4],r=1,l=1)$ and columns $k=5$, $6$ respectively. 

\item[$-$] For the second element of $\mathcal{F}_{\mathcal{R},1}$, i.e., $\{7,8\}$, define  $\mathcal{S}=\mathcal{T}\cup \{7,8\}=\{1,2,3,4,7,8\}$.  By Lemma \ref{lem-factor}, we can divide all the $\tau_1=2$-subset of $\mathcal{S}$ into $\Lambda_{2}={6-1\choose 2-1}=5$ parallels of $\mathcal{S}$, i.e.,  
\begin{align*}
\mathcal{F}_{\mathcal{S},1}& = \{\{1, 2\}, \{3, 4\}, \{7, 8\}\}\ \ 
\mathcal{F}_{\mathcal{S},2} = \{\{1, 3\}, \{2, 7\}, \{4, 8\}\}\ \
\mathcal{F}_{\mathcal{S},3} = \{\{1, 4\}, \{2, 8\}, \{3, 7\}\}\nonumber\\
\mathcal{F}_{\mathcal{S},4} &= \{\{1, 7\}, \{2, 4\}, \{3, 8\}\}\ \
\mathcal{F}_{\mathcal{S},5} = \{\{1, 8\}, \{2, 3\}, \{4, 7\}\}.
%\label{eq-paralles-example-S}
\end{align*}Since $\{7,8\}$ is the third element of the first parallel, i.e., $\mathcal{F}_{\mathcal{S},{\color{red}1}}$, we put the vector
\begin{align*}
\left(\mathcal{S},\left\lceil\frac{{\color{red}1}+(l-1)\Lambda_{2}}{G}\right\rceil\right)=\left([4]\cup\{7,8\},\left\lceil\frac{1+(1-1)\cdot 5}{2}\right\rceil\right)=([4]\cup\{7,8\},1)	
\end{align*}	
in the entries at row $(\mathcal{T}=[4],r=1,l=1)$ and columns $k=7$, $8$ respectively.
\end{itemize}    
Using the above assignment rule, we can assign all non-star entries to the rows $(\mathcal{T}, r=1, l=1)$ where $\mathcal{T}=\{1,2,3,5\},\{1,2,5,6\},\{1,2,4,6\},\{2,4,5,6\},\{3,4,5,6\}$, and thereby obtain the subarray displayed in Fig.~\ref{fig-ex1}.
\begin{figure}[http!]
\centering
\includegraphics[width=6in]{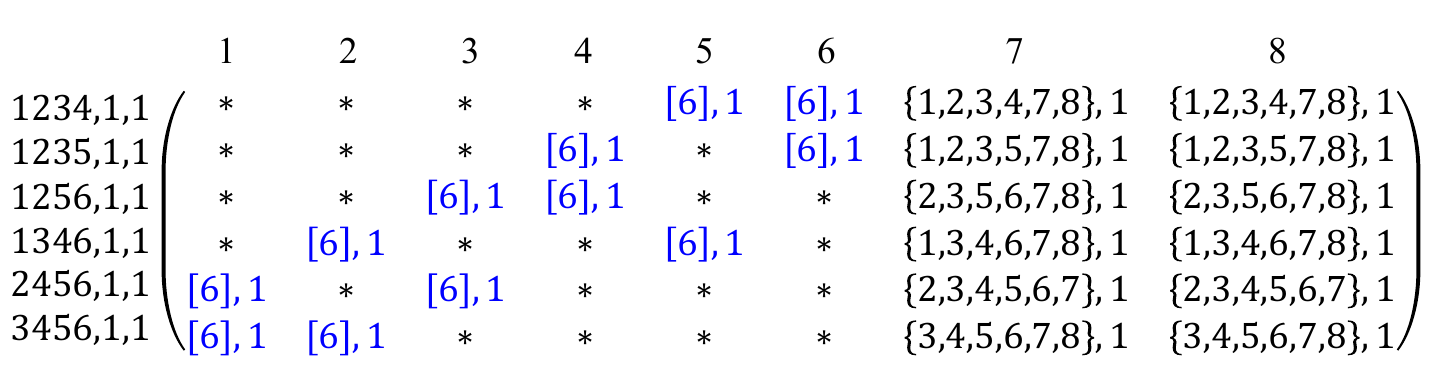}
\caption{The subarray of $\mathbf{B}$ including the row $(\mathcal{T},1,1)$ where $\mathcal{T}=1234$, $1235$, $1256$, $1346$, $2456$, $3456$.}
\label{fig-ex1}
\end{figure} We can check that the vector $([6], 1)$ highlighted in red satisfies conditions C$3$-$4$ of Definition \ref{def-MIMO-PDA}. Repeating this process for all non-star entries generates our desired array $\mathbf{B}$. Due to its $420$ rows, the complete array is omitted; only the representative subarray is provided. We can see that each $6$-subset $\mathcal{S}$ of $[K_1]$ appears in $\mathbf{B}$. In addition, the vectors containing $\mathcal{S}$ is divided into 
$$\left\lceil\frac{{\color{red}\Lambda_{2}}+(\frac{G}{\gcd(G,\Lambda_{2})}-1)\Lambda_{2}}{G}\right\rceil=
\frac{\Lambda_{2}}{\gcd(G,\Lambda_{2})}=\frac{5}{\gcd(2,5)}=5$$ 
different vectors. Hence, there are exactly $S_1=5{K_1\choose t_1+\tau_1}=5{8\choose 6}=140$ different vectors in $\mathbf{B}$. In addition, we have $\mu=\lceil G/{t_1+\tau_1-1\choose t_1}\rceil=\lceil 2/5\rceil=1$. So, $\mathbf{B}$ is a  $(2,3,8,F_1=420,Z_1=210,S_1=140)$ MIMO-PDA with consistency number $\mu=1$.

From \eqref{eq-MN-array}, the first coordinate of the vector entries in the TST MIMO-PDA introduced in \cite{TSJA} coincides with that of the vectors assigned to the non‑star entries of $\mathbf{B}$. Consequently, $\mathbf{B}$ can be obtained by modifying only the second coordinate of each vector in the array $\mathbf{A}$ defined in \eqref{eq-MN-array}. For brevity, the resulting MIMO-PDA $\mathbf{B}$ will be referred to as a new TST MIMO-PDA.

$\bullet$ {\bf Step 2} Construction of a $(2,13,24,7560,3780,2520)$ MIMO-PDA $\mathbf{X}$: We can obtain first an array $\mathbf{Q}$ by replicating $\mathbf{B}$ $m=3$ times horizontally. Clearly $\mathbf{Q}$ has $420$ rows, $mK_1=3\cdot 8=24$ columns and $140$ different vectors; each column has exactly $210$ stars, the subarray $\mathbf{Q}^{(\mathbf{e})}$ of $\mathbf{Q}$ including the rows and columns containing $\mathbf{e}$ can be written as  $\mathbf{Q}^{(\mathbf{e})}=(\mathbf{B}^{(\mathbf{e})}$ $\mathbf{B}^{(\mathbf{e})}$ $\mathbf{B}^{(\mathbf{e})})$ 
where $\mathbf{B}^{(\mathbf{e})}$ is the subarray of $\mathbf{B}$ including the rows and columns containing $\mathbf{e}$. For instance, in Fig. \ref{fig-ex3},  we have the subarray $\mathbf{Q}^{([6],1)}$ which is generated by replicating $\mathbf{B}^{([6],1)}$ $3$ times horizontally. In $\mathbf{Q}^{([6],1)}$, vector $([6],1)$ appears  exactly $G=2$ times in each column, i.e., the condition C$3$ of Definition \ref{def-MIMO-PDA} holds. Moreover, each row has exactly $12$ stars because each row of $\mathbf{B}^{([6],1)}$ has exactly $4$ stars; the star position sets differ across all rows because $\mathbf{B}$ has the consistency number $\mu=1$. Therefore, the condition C$4$ of Definition \ref{def-MIMO-PDA} holds. There are exactly $g_1=3\cdot 12=36$. Then $\mathbf{Q}$ is $36$-$(2,13,24,420,210,140)$ MIMO-PDA.

\begin{figure}[http!]
	\centering
	\includegraphics[width=6in]{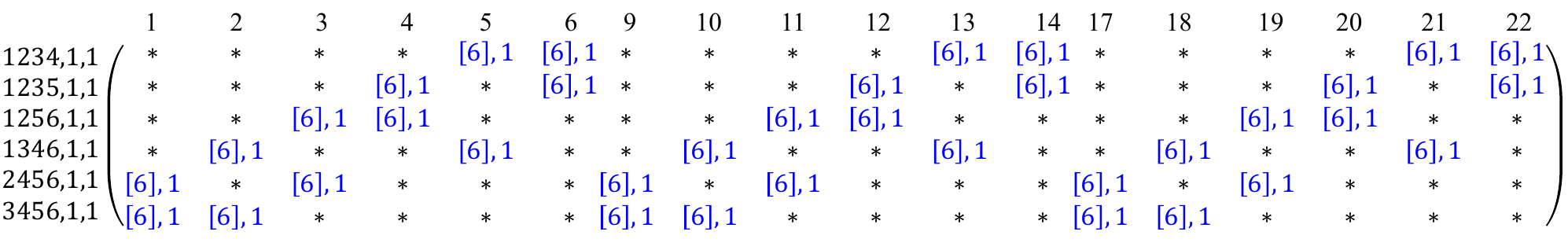}
	\caption{The subarray $\mathbf{Q}^{([6],1)}$ of $\mathbf{Q}$.}
	\label{fig-ex3}
\end{figure}
Recall that the maximum sum-DoF is $G(mt_1+\tau)=2\cdot(3\cdot 4+7) =38>36=g_1$. We aim to append another MIMO-PDA $\mathbf{Y}$ with the same number of columns and a sum-DoF of $g_2=38-36=2$ (detailed in the following step) below $\mathbf{Q}$. To ensure successful concatenation, we obtain our desired array $\mathbf{X}$ by vertically replicating $\mathbf{Q}$ $m\binom{t_1+\tau_1}{\tau_2}=3\binom{4+2}{1}=18$ times and appending the replication index to its each vector. Consequently, all entries across different copies are mutually distinct. For instance, the vector $([6],1)$ in the first replication becomes $([6],1,1)$, in the second replication it becomes $([6],1,2)$, and so forth. Since each replication is also a $36$-$(2,13,24,420,210,140)$ MIMO-PDA with consistency number $\mu=1$, $\mathbf{X}$ is a  $36$-$(2,11,24,F_{\mathbf{X}}=7560,Z_{\mathbf{X}}=3780,S_{\mathbf{X}}=2520)$ MIMO-PDA with the consistency number $\mu=1$.

$\bullet$ {\bf Step 3}  Construction of a $(2,13,24,420,210,2520)$ MIMO-PDA $\mathbf{Y}$: Let $\tau_2=\tau-m\tau_1=7-3\cdot 2=1$. Recall that the sum-DoF $g_1=12$ of $\mathbf{B}$ is generated from its $t_1+\tau_1=6$ columns. To obtain the MIMO-PDA $\mathbf{Y}$ with sum-DoF $g_2=2$, we can reorganize these columns into $(t_1+\tau_1)/\tau_2=6/1=6$ groups each of which has size of $\tau_2$ and generates sum-DoF $2$. For convenience in computing the degree of the bipartite graph introduced later, we first replicate $\mathbf{B}$ 
$\Lambda_3 = \binom{t_1+\tau_1-1}{\tau_1} = \binom{5}{0} = 1$ times. Then, the columns of each replica are partitioned into $\Lambda_3$ parallel classes of this $t_1+\tau_1=6$ column set, each consisting of $\tau_2=1$-subsets that collectively yield the sum-DoF $g_2=2$. For instance,  the vector $([6],1)$ is located in columns $1$ through $6$. By obtaining a parallel $\mathcal{U}_{[6],1}=\{\{1\},\{2\},\{3\},\{4\},\{5\},\{6\}\}$ of $[6]$, we can represent it across these columns as $([6],1,1,1)$, $([6],1,1,2)$, $([6],1,1,3)$, $([6],1,1,4)$, $([6],1,1,5)$, and $([6],1,1,6)$, where the third coordinate denotes the parallel index and the fourth indicates the element’s position within that parallel. Since each subarray $\mathbf{T}^{(\mathbf{e},i,j')}$ of $\mathbf{T}$ is also a subarray $\mathbf{B}^{(\mathbf{e})}$, $\mathbf{T}$ satisfies conditions C$3$-$4$. There are exactly $\binom{t_1+\tau_1-1}{\tau_1-1}(t_1+\tau_1)S_1/\tau_2=6\cdot 140=840$ different vectors in $\mathbf{T}$. Hence, $\mathbf{T}$ is a $g_2=2$-$(2,13,24, 420,210,840)$ MIMO-PDA with consistency number $\mu=1$.

Finally, we obtain our desired array $\mathbf{Y}$ by horizontally replicating $\mathbf{T}$ $m=3$ times and appending the replication index to its each vector. Consequently, all entries across different copies are mutually distinct. For instance, the vector $([6],1,1,1)$ in the first replication becomes $([6],1,1,1,1)$, in the second replication it becomes $([6],1,1,1,2)$, and in the third replication it becomes $([6],1,1,1,3)$. Since each replication is also a a $2$-$(2,13,8, 420,210,840)$ MIMO-PDA with consistency number $\mu=1$, $\mathbf{Y}$ is a  $2$-$(2,13,24,F_{\mathbf{Y}}=420,Z_{\mathbf{Y}}=210,S_{\mathbf{Y}}=2520)$ MIMO-PDA with the consistency number $\mu=1$.

$\bullet$ {\bf Step 4}  Construction of a  $(2,13,24,7980,3990,2520)$ MIMO-PDA $\mathbf{P}$: We will construct an array $\mathbf{P}$ by concatenating $\mathbf{X}$ and $\mathbf{Y}$ vertically. In $\mathbf{P}$, each column contains $Z =Z_{\mathbf{X}}+Z_{\mathbf{Y}}=3780+210=3990$ stars, and the total number of subpacketization is $F =F_{\mathbf{X}}+F_{\mathbf{Y}}=7560+420=7980$. In the following, we will apply the saturating matching argument from Lemma \ref{le-hall-theorem} to adjust the vectors in $\mathbf{Y}$.  

Let $\mathcal{X}$ and $\mathcal{Y}$ denote the sets of distinct vectors appearing in the arrays $\mathbf{X}$ and $\mathbf{Y}$ respectively. We can obtain a $\mathcal{G}_{\mathbf{X},\mathbf{Y}}=(\mathcal{X},\mathcal{Y},\mathcal{E})$ where a vertex $\mathbf{x}=(\mathcal{S},a,x)\in\mathcal{X}$ is adjacent with a vertex $\mathbf{y}=(\mathcal{S}',a',y_1,y_2,y_3)\in\mathcal{Y}$ if and only if  $|\mathcal{S}\cup\mathcal{S}'|=t_1+\tau_1+\tau_2$, $\mathcal{S}'\setminus\mathcal{S}=\mathcal{U}_{\mathcal{S}',y_1}(y_2)$, and $2=G\leq \langle(a-a')G\rangle_{\Lambda_{2}=5}\leq \Lambda_{2}-G=3$. Here, the notation $\mathcal{U}_{\mathcal{S}',y_1}(y_2)$ represents the $y_2$-th element of the $y_1$-th parallel of $\mathcal{S}'$. Let us consider the vertices
\begin{align*}
\mathbf{x}&=(\mathcal{T},a,i)=([6],1,1),\ ([5]\cup\{7\},1,1),\ ([5]\cup\{8\},1,1)\\
\mathbf{y}&=(\mathcal{S}',a',y_1,y_2,y_3)=([5]\cup\{7\},2,1,6,1),\ ([5]\cup\{8\},2,1,6,1),\ ([6],2,1,6,1)
\end{align*}respectively. It is clear that $|[6] \cup ([5] \cup \{7\})| = t_1 + \tau_1 + \tau_2 = 7$; the element $ ([5] \cup \{7\}) \setminus [6]= \{7\}$ is the $6$-th element in the parallel $\mathcal{U}_{[5]\cup\{7\},1}=\{\{1\},\{2\},\{3\},\{4\},\{5\},\{7\}\}$ of $[5]\cup\{7\}$, i.e., $\mathcal{U}_{[5]\cup{7},1}(6)$; $2\leq \langle(a-a')G\rangle_{\Lambda_{2}=5}=\langle(1-2)\cdot 2\rangle_{5}=3\leq 3$. By our edge difinition, the vertex $([6],1,1)$ is adjacent with the vertex $([5]\cup\{7\},2,1,6,1)$ in $\mathcal{G}_{\mathbf{X},\mathbf{Y}}$. Similarly, vertex $([6],1,1)$ is adjacent to $([5]\cup\{8\},2,1,6,1)$; vertex $([5]\cup\{7\},1,1)$ is adjacent to both $([5]\cup\{8\},2,1,6,1)$ and $([6],2,1,6,1)$; and vertex $([5]\cup\{8\},1,1)$ is adjacent to $([5]\cup\{7\},2,1,6,1)$ and $([6],2,1,6,1)$. Then, the subgraph of $\mathcal{G}_{\mathbf{X},\mathbf{Y}}$ in Fig. \ref{fig-ex5} can be obtained. 
\begin{figure}[http!]
	\centering
	\includegraphics[width=5in]{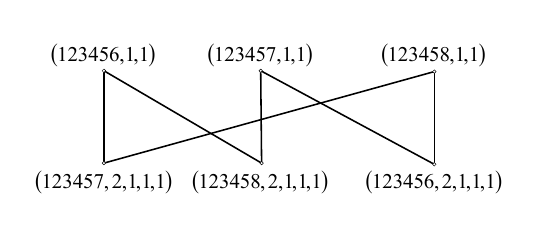}
	\caption{The induced subgraph of $\mathcal{G}_{\mathbf{X},\mathbf{Y}}$ generated by six vertex}.
	\label{fig-ex5}
\end{figure}

Let us consider the subarray $\mathbf{P}'$ which including the rows and columns containing the vertices at edge $(([6],1,1),([5]\cup\{7\},2,1,6,1))$, as listed in Fig. \ref{fig-ex4}.
\begin{figure}[http!]
	\centering
	\includegraphics[width=6in]{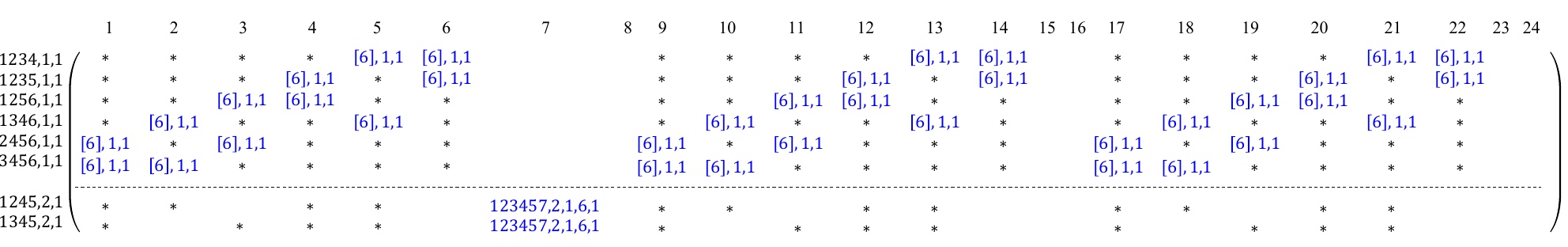}
	\caption{The subarray $\mathbf{P}'$ }.
	\label{fig-ex4}
\end{figure}

 Replacing the vertex $([5]\cup\{7\},2,1,1,1)$ with $([6],1,1)$, the modified $\mathbf{P}'$ contains exactly $7$ non-star entries per row and two occurrences of $([6],1,1)$ per column, while also satisfying Condition~C$4$. Therefore, given a perfect matching in $\mathcal{G}_{\mathbf{X},\mathbf{Y}}$, we may replace each vertex from $\mathcal{Y}$ with its matched vertex from $\mathcal{X}$, thereby completing the adjustment of the vectors in $\mathbf{Y}$.  

In the following, it is sufficient to show that there exists a perfect matching on $\mathcal{G}_{\mathbf{X},\mathbf{Y}}$. By Lemma \ref{le-hall-theorem}, we should consider the degree d$(\mathcal{X})$ and d$(\mathcal{Y})$ respectively.

Let us consider the degree of $([6],1,1)$. There are exactly $12={K_1-6\choose \tau_2}{6\choose t_2}={2\choose 1}{6\choose 1}$ distinct $t_1+\tau_1=6$-subset $\mathcal{S}'$ satisfying $|\mathcal{S}\cup\mathcal{S}'|=t_1+\tau_1+\tau_2$, i.e., 
\begin{align*}
& [5]\cup \{7\}, [5]\cup \{8\}, [4]\cup \{6,7\}, [4]\cup \{6,8\}, [3]\cup \{5,6,7\}, [4]\cup \{5,6,8\},\\
& [2]\cup \{4,5,6,7\}, [4]\cup \{4,5,6,8\}, \{1\}\cup \{3,4,5,6,7\},\{1\}\cup \{3,4,5,6,8\}, [2:7],[2:6]\cup \{8\}.
\end{align*}For each $\mathbf{S}'$, there is exactly one  parallel in which every element is a singleton, and the position of the $\tau_2=1$-subset $\mathcal{S}'\setminus\mathcal{S}$ is fixed in this parallel.

Since $\mathbf{T}$ is replicated exactly $m=3$ times, and given the integer $a=1$, there are exactly $2=\frac{\Lambda_{2}-2G}{\gcd(\Lambda_{2},G)}+1$ $a'$ satisfying $2\leq \langle(a-a')G\rangle_{5}\leq 3$. Consequently, the degree of $([6],1,1)$ equals $12\cdot 1\cdot 3\cdot 2=72$. By the same computing method, every other vertex in $\mathcal{X}$ also has degree $72$. 

Conversely, for any vertex $\mathbf{y}$, e.g., $([5]\cup{7}, a', y_1=1, y_2=6, c)$, we have $\mathcal{U}_{[5]\cup{7},1}(6)={7}$. There are exactly $2 = \binom{K_1-6}{\tau_2} = \binom{2}{1}$ subsets $\mathcal{S}$ satisfying $|\mathcal{S} \cup \mathcal{S}'| = t_1 + \tau_1 + \tau_2$, namely $[6]$ and $[5]\cup\{8\}$. Since $\mathbf{X}$ is obtained by replicating $\mathbf{Q}$ $18$ times, and for a given integer $a'=1$, there are exactly $2 = \frac{\Lambda_{2} - 2G}{\gcd(\Lambda_{2}, G)} + 1$ values of $a$ such that $2 \le \langle (a-a')G \rangle_5 \le 3$. It  follows that the degree of $\mathbf{y}$ equals $2 \times 18 \times 2 = 72$. Using the same computing method, every other vertex in $\mathcal{Y}$ also has degree $72$. We have d$(\mathcal{X})=$d$(\mathcal{Y})$. 

Finally, by using a perfect matching to adjust the vectors in $\mathbf{Y}$, we obtain the desired $38-(2,13,24,7980,3990,2520)$ MIMO-PDA $\mathbf{P}$ with the sum-DoF $g=g_1+g_2=36+2=38$ and the consistency number $\mu=1$.

Based on $\mathbf{P}$ and by Theorem \ref{th-Fundamental}, we can obtain a 
$(G,L,mK_1,M,N)=(2,13,24,M,N)$ MIMO coded caching scheme with the memory ratio $M/N=0.5$, the optimal sum-DoF $38$, and the subpacketization $F=7980$. In this example, the TST MIMO-PDA in \cite{MMA} also achieves the optimal sum-DoF $38$ while having the subpacketization $1784742960$. Clearly our subpacketization is much smaller than that of the TST MIMO-PDA.

\section{CONCLUSION}
In this paper, we first introduced a MIMO-PDA to characterize the placement and delivery strategies simultaneously of a $(G,L,K,M,N)$ MIMO coded caching scheme under the uncoded placement and one-shot ZF delivery. It is worth noting that all the existing schemes can be represented by MIMO-PDA. Then the upper bound on sum-DoF which is consistent with the optimal problem proposed by \cite{MMA} was derived. According to our upper bound, we obtained two classes of achievable MIMO-PDAs. Finally, by theoretical and numerical comparisons, we showed that our schemes significantly reduced the subpacketization compared to the existing schemes.

\begin{appendices}

\section{The proof of Rank$(\mathbf{V}_i)= g_i$}
\label{app-proof-rank}
We will use the following Schwartz-Zippel (SZ) Lemma in \cite{DL} to show Rank$(\mathbf{V}^{(s)}_i)= g_i$ where $i\in[r]$. We first consider the case that $\bar{\mathbf{H}}^{(s)}_{i,j}$ where $j\in[l_{i,j}-1]$ in \eqref{eq-unique-vector} has exactly $L-\rho$. 

In the following, we first show that given a full rank $\mathbf{V}^{(s)}_i$, we can find $r$ non-zero channel matrices $\mathbf{H}_{k_i}$, $i\in[r]$ such that \eqref{eq-unique-vector} is satisfied for each $j\in[g_i]$. We assume that $\mathbf{V}^{(s)}_i(\cdot,j)$ is the $j$-th unit vector, i.e., with $1$ at the $j$-th position and $0$ elsewhere. Then, from \eqref{eq-unique-vector} we have 
\begin{align*}
\bar{\mathbf{H}}^{(s)}_{i,1}(\cdot,1)=\mathbf{0},\ \bar{\mathbf{H}}^{(s)}_{i,2}(\cdot,2)=\mathbf{0},\ \ldots,\  \bar{\mathbf{H}}^{(s)}_{i,g_i}(\cdot,g_i)=\mathbf{0}.
\end{align*}This implies that 
\begin{align}
\label{eq-zero} 
\begin{cases}
\mathbf{H}^{(s)}_{i,1,1}(\cdot,1)=\mathbf{0}\\
\mathbf{H}^{(s)}_{i,1,2}(\cdot,1)=\mathbf{0}\\
\ \ \ \ \ \ \ \ \ \vdots\\
\mathbf{H}^{(s)}_{i,1,l_{i,1}-1}(\cdot,1)=\mathbf{0}
\end{cases}\ \ 
\begin{cases}
	\mathbf{H}^{(s)}_{i,2,1}(\cdot,2)=\mathbf{0}\\
	\mathbf{H}^{(s)}_{i,2,2}(\cdot,2)=\mathbf{0}\\
\ \ \ \ \ \ \ \ \ 	\vdots\\
	\mathbf{H}^{(s)}_{i,2,l_{i,2}-1}(\cdot,2)=\mathbf{0}
\end{cases}\ \ \cdots\ \ \ 
\begin{cases}
	\mathbf{H}^{(s)}_{i,g_i,1}(\cdot,g_i)=\mathbf{0}\\
	\mathbf{H}^{(s)}_{i,g_i,2}(\cdot,g_i)=\mathbf{0}\\
\ \ \ \ \ \ \ \ \ 	\vdots\\
	\mathbf{H}^{(s)}_{i,g_i,l_{i,g_i}-1}(\cdot,g_i)=\mathbf{0}
\end{cases}.
\end{align}By condition C$4$-b), i.e., at most $\rho$ rows contain the same non‑star position set, the channel matrix $\mathbf{H}^{(s)}_{i,j}$ in \eqref{eq-zero}, where $j\in[g_i]$ has at most $\rho$ zero columns; without loss of generality we assume $\mathbf{H}^{(s)}{i,j}(\cdot,[\rho])=\mathbf{0}$.
Consider now the submatrices $\mathbf{H}^{(s)}_{i,j}(\cdot,[\rho+1:g_i])$ and $\mathbf{H}^{(s)}_{i,j}(\cdot,[g_i+1:L])$, which together form the square $\mathbf{H}^{(s)}_{i,j}(\cdot,[\rho+1:L])$. In addition, by condition C$4$-b) each column of $\mathbf{H}^{(s)}_{i,j}(\cdot,[\rho+1:g_i])$ possesses at least $G$ consecutive entries that can be assigned nonzero i.i.d. values from $\mathbb{C}$, while all entries of $\mathbf{H}^{(s)}_{i,j}(\cdot,[g_i+1:L])$ are likewise chosen i.i.d. from $\mathbb{C}$. Consequently, the combined square $\mathbf{H}^{(s)}_{i,j}(\cdot,[\rho+1:L])$ is invertible. In other word, the $j$-th column of matrix $\bar{\mathbf{H}}^{(s)}_{i,j}$ is a zero column, and $\bar{\mathbf{H}}^{(s)}_{i,j}$ has a full-rank submatrix $\bar{\mathbf{H}}^{(s)}_{i,j}(\cdot,[\rho+1:L])$, satisfying $\bar{\mathbf{H}}^{(s)}_{i,j}\mathbf{V}^{(s)}_{i}(\cdot,j)=\mathbf{0}$. Then, we have  $\mathbf{V}^{(s)}([g_i],[g_i])$ is a unit matrix. So, $\mathbf{V}^{(s)}_{i}$ is not a  non-zero polynomial, by the following SZ Lemma \cite{DL}, the matrix $\mathbf{V}^{(i)}$ is full column rank, i.e., Rank$(\mathbf{V}_{i})= g_i$. Then our proof is complete. 
\begin{lemma}[Schwartz, Zippel \cite{DL}]
	\label{le-Sc-Zi}
	Let $f\in \mathcal{F}[x_1,x_2,\ldots,x_n]$ be a non-zero polynomial of total degree $d\geq0$ over a field $\mathcal{F}$. Let $\mathcal{S}$ be a finite subset of $\mathcal{F}$ and let $r_1$, $r_2$, $\ldots$, $r_n$ be selected at random independently and uniformly from $\mathcal{S}$. Then, Pr$(f(r_0,r_1,\ldots,r_{n-1})=0)\leq \frac{d}{|\mathcal{S}|}$.
\end{lemma}

\section{proof of Theorem \ref{th-DoF}}\label{proof-DoF}
Given a $(G,L,K,F,Z,S)$ MIMO-PDA $\mathbf{P}$, let's assume that each integer $s\in[S]$ occurs $g^{(s)}=g^{(s)}_1+g^{(s)}_2+\cdots+g^{(s)}_{r_s}$ times in $\mathbf{P}$, from Condition C$3$ of Definition \ref{def-MIMO-PDA} we know that where $g_i\leq G, i\in[r_s]$, denoted by 
\begin{align*} 
\mathbf{P}(f_{1,1},k_{1}), \ldots,\mathbf{P}(f_{1,g_1},k_{1}), ,\mathbf{P}(f_{2,1},k_{2}), \ldots,,\mathbf{P}(f_{r,g_{r_s}},k_{r}). 
\end{align*} We can obtain the subarray $\mathbf{P}^{(s)}$ with $r_s$ columns, and  let $f_{i,j}$ and $k_i$, $j\in[g_i]$, $i\in [r]$ represent the row indices and column indices of $\mathbf{P}^{(s)}$, respectively. We say that a star is used by an integer $s\in[S]$ if it appears in a subarray $\mathbf{P}^{(s)}$ of $\mathbf{P}$. For each subarray $\mathbf{P}^{(s)}$, $s\in[S]$, we assume that there are $e_{i,j}$ integer entries in the row $f_{i,j}\in[F]$, and there are $r_s-e_{i,j}$ stars used by the integer $s$ of $\mathbf{P}^{(s)}(f_{i,j},k_i)$. Then the number of stars used by all the integer $s$'s in $\mathbf{P}^{(s)}$ is exactly $\sum_{i=1}^{r_s}\sum_{j=1}^{g_i^{(s)}}(r_s-e_{i,j})$, and the total number of stars used in all $\mathbf{P}^{(s)}$, $s\in [S]$, is
\begin{eqnarray*} M&=&\sum_{s=1}^S\sum_{i=1}^{r_s}\sum_{j=1}^{g_i^{(s)}}(r_s-e_{i,j})=\sum_{s=1}^Sr_s(g_1^{(s)}+g_2^{(s)}+\cdots +g_{r_s}^{(s)} )-\sum_{s=1}^S\sum_{i=1}^{r_s}\sum_{j=1}^{g_i^{(s)}}e_{i,j}\\
&=&\sum_{s=1}^Sr_sg^{(s)} -\sum_{s=1}^S\sum_{i=1}^{r_s}\sum_{j=1}^{g_i^{(s)}}e_{i,j}.
\end{eqnarray*}
Next, we consider the array $\mathbf{P}$ and assume that each row $j\in [F]$ has $e'_j$ integer entries, then the times of all stars used by the integer entries in $j$-th row is at most $e'_j(K-e'_j)$. So the total times of all stars used in $\mathbf{P}$ is at most $M'=\sum_{j=1}^Fe'_j(K-e'_j)$. Clearly, $M \leq M'$, i.e.,
\begin{eqnarray}\label{num-star}
\sum_{s=1}^Sr_sg^{(s)} -\sum_{s=1}^S\sum_{i=1}^{r_s}\sum_{j=1}^{g_i^{(s)}}e_{i,j} \leq \sum_{j=1}^Fe'_j(K-e'_j)\label{1}.
\end{eqnarray}
Since $n=(F-Z)K$ is the total number of integers in $\mathbf{P}$, we have $n=\sum_{s=1}^Sg^{(s)} =\sum_{j=1}^F e'_j$. Based on Condition C$3$ of Definition \ref{def-MIMO-PDA} we have
$$\sum_{s=1}^Sr_sg^{(s)} \geq \sum_{s=1}^S\frac{g^{(s)}}{G}g^{(s)}=\frac{1}{G}\sum_{s=1}^S(g^{(s)})^2.$$
Thus from \eqref{num-star} we get
\begin{eqnarray*}
\frac{1}{G}\sum_{s=1}^S(g^{(s)})^2 +\sum_{j=1}^{F}(e'_j)^2\leq \sum_{s=1}^S\sum_{i=1}^{r_s}\sum_{j=1}^{g_i^{(s)}}e_{i,j} +\sum_{j=1}^{F}Ke'_j.
\end{eqnarray*}
In addition, by the convexity and
$\sum_{i=1}^{r_s}\sum_{j=1}^{g_i^{(s)}}e_{i,j}\leq g^{(s)}\lceil\frac{L}{G}\rceil$ from C$4$ of Definition \ref{def-MIMO-PDA}, we can obtain
$$\sum_{s=1}^{S}(g^{(s)})^2\geq \frac{1}{S}(\sum_{s=1}^{S}g^{(s)})^2=\frac{n^2}{S},\ \ \ \ \sum_{j=1}^{F}(e'_j)^2\geq \frac{1}{F}(\sum_{s=1}^{S}e'_j)^2=\frac{n^2}{F}.$$
Then
\begin{eqnarray*}
\frac{n^2}{GS}+\frac{n^2}{F}\leq \sum_{s=1}^{S}g^{(s)}\left\lceil\frac{L}{G}\right\rceil+nK = n\left\lceil\frac{L}{G}\right\rceil+nK,
\end{eqnarray*} i.e., $S\geq \frac{nF}{FG\lceil\frac{L}{G}\rceil+GKZ}$.
Thus we get the sum-DoF of $\frac{K(F-Z)}{S}\leq \frac{FG\lceil\frac{L}{G}\rceil+GKZ}{F}=\frac{GKZ}{F}+G\lceil\frac{L}{G}\rceil$, where the equation holds if and only if $e_{i,j}=\lceil\frac{L}{G}\rceil$, $g^{(1)}=g^{(2)}=\cdots=g^{(S)}=\frac{n}{S}$ and $e'_1=e'_2=\cdots=e'_F$. Then the proof is completed.

\section{The proof of Theorem \ref{th-latin}}	
\label{app-proof-latin}
According to the aforementioned constructing idea, which includes placing the stars under cyclic wrap-around topology in each row of a square with size $K$,  replicating this square vertically $G$ times and adjusting the integers in each column, the obtained $KG \times K$ array $\mathbf{A}=(\mathbf{A}((j,i),k))_{j,k\in [K], i\in[G]}$ can be represented as follows. For any integers $j,k\in[K]$ and $g\in [G]$, the entry $a_{(j,i),k}$ is defined according to the following rule: 
\begin{align}
	\label{eq-Square-array}
	\mathbf{A}((j,i),k)=
	\begin{cases}
		* & \text{if } k \in \mathcal{I}_{j},\\[2pt]
		\Big\lceil \frac{\lambda_{(j,i),k}}{G}\Big\rceil & \text{otherwise},
	\end{cases}
\end{align}
where $\mathcal{I}_{j}=\{ j,\langle j+1\rangle_{K},\ldots,\langle j+t-1\rangle_{K} \},
$ and $\lambda_{(j,i),k}$ denotes the order (from top to down) of the non-star entry in the $k$-th column.

Now, let us check the conditions of Definition \ref{def-MIMO-PDA}.  From \eqref{eq-Square-array}, each column has exactly $Z=Gt$ stars; there are exactly $S=K-t$ different integers in A since each column has exactly $G(K-t)$ non-star entries;  Condition C$3$ and C$4$-a) always hold since each integer occurs exactly $G$ times in each column. In addition, each integer occurs at most $K$ consecutive rows of $\mathbf{A}$ and occurs in all the columns of  $\mathbf{A}$. For any two different integers, $j_1,j_2\in[K]$, we have $\mathcal{I}_{j_1}\neq\mathcal{I}_{j_2}$ which implies that any $K$ consecutive rows of $\mathbf{A}$ have different integer entry position sets. Then, the Condition C$4$-b) holds. So, $\mathbf{A}$ is a $(G,L,K,GK,Gt,K-t)$ MIMO-PDA. The proof is completed.

\section{Proof of Theorem \ref{th-grouping}}
\label{proof-grouping}
Assume that $\mathbf{Q}$ is an optimal $(G,L_1,K_1,F,Z,S)$ MIMO-PDA with consistency number $\mu$ and satisfying that each row has exactly $K_1Z_1/F_1$ stars. For any positive integers $m$ and $L$ satisfying $m\lceil L_1/G\rceil=\lceil L/G\rceil$ and $\mu\leq\langle L\rangle_G$, we will construct an optimal $(G,L,K=mK_1,F,Z,S)$ MIMO-PDA $\mathbf{P}$ for any positive integers $m$ and $L$  by replicating $\mathbf{Q}$ $m$ times horizontally. It is straightforward to see that $\mathbf{P}$ satisfies Conditions C$1$, C$2$, C$3$ of Definition \ref{def-MIMO-PDA}. We then focus on Condition C$4$ of Definition \ref{def-MIMO-PDA}.

Finally, we verify the condition C$4$ of Definition \ref{def-MIMO-PDA}. First, we claim that each row of $\mathbf{Q}$ contains exactly $t_1=K_1Z_1/F_1$ stars. Suppose, to the contrary, that some row has fewer than $t$ stars. Consider an integer $s$ appearing in that row. By the property of MIMO-PDA, $s$ occurs exactly $G(t_1+\lceil L_1/G\rceil)$ times in $\mathbf{Q}$, and each column containing $s$ exactly $G$ times by condition C$3$. Hence $s$ must appers in at least $t_1+\lceil L_1/G\rceil$ columns. Moreover, every row of $\mathbf{P}^{(s)}$ has at most $\lceil L_1/G\rceil$ integer entries, which implies that each row of $\mathbf{P}^{(s)}$ has at least $t_1$ star entries. This contradicts the assumption that the original row had fewer than $t_1$ stars. Then, $\mathbf{Q}^{(s)}$ has exactly $t_1+\lceil L_1/G\rceil$ columns where each row of $\mathbf{Q}^{(s)}$ has $\lceil L_1/G\rceil$ integer entries and $t_1$ stars. So in the subarray of $\mathbf{P}$, which includes the rows and columns containing $s$, has exactly $m(t_1+\lceil L_1/G\rceil)$ columns where each row has $m\lceil L_1/G\rceil$ integer entries and $mt_1$ stars. Hence, $\mathbf{P}$ satisfies Condition C$4$ of Definition \ref{def-MIMO-PDA}. This implies that $\mathbf{P}$ is an $(G,L,mK_1,F_1,Z_1,S_1)$ MIMO-PDA. 

Recall that $\lceil L/G\rceil=m\lceil L_1/G\rceil$, the sum-DoF of $\mathbf{P}$ is $m(Gt+G\lceil L_1/G\rceil)=G(mt+\lceil L/G\rceil)$, which is the objective sum-DoF in Theorem~\ref{th-DoF}.

\section{Proof of Theorem~\ref{th-MAPDA-hypergraph}	}
\label{proof-regular}
For any positive integers $K_1$, $t_1$, $G$, $L_1$, and $L$, let $\tau_1=\lceil L_1/G\rceil$, $\tau=\lceil L/G\rceil$, $m=\lfloor \frac{\tau}{\tau_1}\rfloor\geq 1$, and $\tau_2=\tau-m\tau_1$. Assume that $\tau_1|\{t_1$, $K_1\}$, $t_1+\tau_1<K_1$, $\tau_2|(t_1+\tau_1)$, and $2G\leq {t_1+\tau_1-1\choose t_1}$. We will construct an optimal $(G$, $L$, $K=mK_1$, $F=\frac{G\Lambda_{1}\Lambda_{3}}{\gcd(G,\Lambda_{2})}\left(\frac{m(t_1+\tau_1)}{\tau_2}+1 \right){K_1\choose t_1}$,
$Z=\frac{G\Lambda_{1}\Lambda_{3}}{\gcd(G,\Lambda_{2})}\left(\frac{m(t_1+\tau_1)}{\tau_2}+1 \right){K_1-1\choose t_1-1}$,
$S=\frac{m\Lambda_{2}\Lambda_{3}}{\gcd(G,\Lambda_{2})}\frac{t_1+\tau_1}{\tau_2}{K_1\choose t_1+ \tau_1})$ MIMO-PDA where $\Lambda_{1}={K_1-t_1-1\choose  \tau_1-1}$,  $\Lambda_{2}={t_1+\tau_1-1\choose  \tau_1-1}$, and $\Lambda_{3}={t_1+\tau_1-1\choose  \tau_2-1}$.

The main idea is as follows. Using the Baranyai's Theorem \cite{Baranyai1975}, we first construct a new TST MIMO-PDA $\mathbf{B}$ by reputing a new integer in the second coordinate of each non-star entry of TSTS MIMO-PDA defined in \eqref{eq-MN-array}. By replicating $\mathbf{B}$ horizontally and vertically with appropriate integer adjustments, we then construct two MIMO-PDAs, $\mathbf{X}$ and $\mathbf{Y}$, characterized by the same number of users $K$ and the same integer set $[S]$, but with distinct sum-DoF values: $g_{\mathbf{X}} = mG\left(t_1 + \tau_1 \right)$ and $g_{\mathbf{Y}} = G\tau_2$, respectively. Finally, using the perfect matching Lemma in \cite{BJM} we obtain the desired $g=g_{\mathbf{X}}+g_{\mathbf{Y}}=G(mt_1+\tau)$-MIMO-PDA $\mathbf{P}$ by vertically concatenating $\mathbf{X}$ and $\mathbf{Y}$ and adjusting the integers in $\mathbf{Y}$. Hence, our construction consists of the following four subsections.   
\subsection{Construction of a new TST MMIMO-PDA $\mathbf{B}$}
\label{sub-new-TSI} 

For any $(K_1 - t_1)$-subset $\mathcal{R}$, by setting $\alpha = \tau_1$ and $v=|\mathcal{R}|$ in Lemma~\ref{lem-factor}, we have $\Lambda_{1} = \binom{K_1 - t_1 - 1}{\tau_1 - 1}$ 
parallel classes on $\mathcal{R}$, denoted $\mathcal{F}_{\mathcal{R},1}, \mathcal{F}_{\mathcal{R},2}, \dots, \mathcal{F}_{\mathcal{R},\Lambda_{1}}$. Similarly, for any $(\tau_1 + t_1)$-subset $\mathcal{S}$, by Lemma~\ref{lem-factor} with $\alpha = \tau_1$ and $v=|\mathcal{S}|$, we also have  $\Lambda_{2} = \binom{\tau_1 + t_1 - 1}{\tau_1 - 1}$ parallel classes on $\mathcal{S}$, denoted $\mathcal{F}_{\mathcal{S},1}, \mathcal{F}_{\mathcal{S},2}, \dots, \mathcal{F}_{\mathcal{S},\Lambda_{2}}$. Using these two families,  we construct a new $g_{1}=G(t_{1}+\tau_{1})$--$(G,L_{1},K_{1},F_{1},Z_{1},S_{1})$ TST MIMO-PDA  
$\mathbf{B}=(\mathbf{B}((\mathcal{T},r,l),k))$, where the parameters are  
\begin{align*}
F_{1}= \binom{K_{1}}{t_{1}}\frac{G\Lambda_{1}}{\gcd(G,\Lambda_{2})},\ Z_{1}= \binom{K_{1}-1}{t_{1}-1}\frac{G\Lambda_{1}}{\gcd(G,\Lambda_{2})},\ 
S_{1}=\frac{\Lambda_{2}}{\gcd(G,\Lambda_{2})}\binom{K_{1}}{t_{1}+\tau_{1}},	
\end{align*}the indexing sets are $\mathcal{T}\in\binom{[K_{1}]}{t_{1}}$, $r\in[\Lambda_{1}]$, $l\in[\frac{G}{\gcd(G,\Lambda_{2})}]$, and $k\in[K_{1}]$, and the entries of $\mathbf{B}$ are defined as follows.  
\begin{align}
	\label{eq-MN-array-hyger}
	\mathbf{B}((\mathcal{T},r,l),k)=
	\begin{cases}
		*, & \text{if } k\in\mathcal{T},\\[4pt]
\bigl(\mathcal{S}, \lceil\frac{d+(l-1)\Lambda_{2}}{G}\rceil\bigr), & \text{if } k\in\mathcal{F}_{\mathcal{R},r}(j) \text{ for some } j,
\end{cases}
\end{align}
where $\mathcal{F}_{\mathcal{R},r}(j)$ denotes the $j$-th element of the $r$-th parallel class of the set $\mathcal{R}=[K_{1}]\setminus\mathcal{T}$, each being a $\tau_{1}$-subset, $\mathcal{S}:=\mathcal{T}\cup\mathcal{F}_{\mathcal{R},r}(j)$, and the integer $d\in[\Lambda_{2}]$ is uniquely determined by the condition $\mathcal{F}_{\mathcal{R},r}(j)\in\mathcal{F}_{\mathcal{S},d}$.  

For instance, when $K_1=8$, $t_1 =4$, $L_1 = 3$, and $G =2$, we have $\tau_1 =\lceil L_1/G\rceil=2$, $t_1 + \tau_1 = 6$ and $K_1-\tau_1=4$ which satisfy the conditions in Lemma \ref{lem-factor}. Let us consider the entries at the row $(\mathcal{T}=[4],r=1,l=1)$ and $k=5$. Clearly, the entry $\mathbf{B}(([4],1,1),5)$ is a non-star entry since $5\not\in[4]$ holds. We have  $\mathcal{R}=[K_1]\setminus\mathcal{T}=[8]\setminus[4]=[5:8]$. By Lemma \ref{lem-factor}, we can divide all the $\tau_1=2$-subset of $\mathcal{R}$ into $\Lambda_{1}={8-4-1\choose 2-1}=3$ parallels of $\mathcal{R}$ in \eqref{eq-paralles-example-R}, where $\mathcal{F}_{\mathcal{R},r=1} = \{\{5, 6\}, \{7, 8\}\}$. Clearly, $k=5$ is the first element of $\mathcal{F}_{\mathcal{R},r=1}$, i.e., $j=1$. Then $\mathcal{S}=\mathcal{T}\cup \mathcal{F}_{\mathcal{R},r=1}(1)=[4]\cup[5:6]=[6]$. By Lemma \ref{lem-factor},  we can divide all the $\tau_1=2$-subset of $\mathcal{S}$ into $\Lambda_{2}={6-1\choose 2-1}=5$ parallels of $\mathcal{S}$ in \eqref{eq-paralles-example-S}, where $\mathcal{F}_{\mathcal{S},1}=\{\{1, 2\}, \{3, 4\}, \{5, 6\}\}$. Clearly, $\mathcal{F}_{\mathcal{R},r=1}(1)=\{5,6\}\in\mathcal{F}_{\mathcal{S},1}$. This implies $d=1$. Hence, the entry $\mathbf{B}(([4],1,1),5)=([6],\lceil\frac{1+(1-1)\cdot 5}{2}\rceil)=([6],1)$.

\subsection{Construction of a MIMO-PDA $\mathbf{X}$}
\label{sub-first-array}
By the proof of Theorem \ref{th-grouping}, for any positive integers $m$, we can obtain an optimal $mg_1$-$(G$, $L'$, $K=mK_1$, $F_1$, $Z_1$, $S_1)$ TST MIMO-PDA $\mathbf{Q}$ where $m\tau_1=\lceil\frac{L'}{G}\rceil$ and where the consistency number $\mu=\lceil G/{t_1+\tau_1-1\choose t_1}\rceil\leq \langle L'\rangle_G$. 

%That is, we can obtain a $F_1\times mK_1$ array $\mathbf{Q}=(q_{(\mathcal{T},r,l),k})$ where $\mathcal{T}\in{[K_1]\choose t_1}$, $r\in[\Lambda_{1}]$, $l\in[\frac{G}{gcd(G,\Lambda_{2})}]$, $k\in [mK_1]$, and the entry $q_{(\mathcal{T},r,l),k}$ is defined as follows.
%\begin{align}
%	\label{eq-array-Q}
%	q_{(\mathcal{T},r,l),k}=\left\{
%	\begin{array}{cc}
%		*&\ \ \ \ \ \text{if}\ b_{(\mathcal{T},r),\langle k\rangle_{K_1}}=*\\
%		b_{(\mathcal{T},r,l),\langle k\rangle_{K_1}}&\text{otherwise}.
%	\end{array}
%	\right.
%\end{align}
In order to make the set of vectors in $\mathbf{X}$ the same as that in $\mathbf{Y}$ constructed in the following subsection,
we construct $\mathbf{X}$ by replicating $\mathbf{Q}$ $m{t_1+\tau_1\choose  \tau_2}$ times vertically and then appending the replication index to its each vector. In other words, our desired $m\binom{t_1+\tau_1}{\tau_2}F_1\times mK_1$ array $\mathbf{X}$ can be defined as
\begin{equation}
	\label{eq-P-Q}
	\mathbf{X}=
	\begin{pmatrix}
		\mathbf{Q}_1 \\
		\mathbf{Q}_2 \\
		\vdots \\
		\mathbf{Q}_{m\binom{t_1+\tau_1}{\tau_2}}
	\end{pmatrix},
\end{equation}
where for each $x\in\bigl[ m\binom{t_1+\tau_1}{\tau_2} \bigr]$, the sub‑array $\mathbf{Q}_{x}=\mathbf{Q}_{x}((\mathcal{T},r,l),k)$ has the indexing sets  $\mathcal{T}\in\binom{[K_1]}{t_1}$, $r\in[\Lambda_{1}]$, $l\in\bigl[\frac{G}{\gcd(G,\Lambda_{2})}\bigr]$, $k\in[mK_1]$, and its entries are given by
\begin{equation}
	\label{eq-array-i}
	\mathbf{Q}_{x}((\mathcal{T},r,l),k)=
	\begin{cases}
		*, & \text{if } 	\mathbf{B}((\mathcal{T},r,l),\langle k\rangle_{K_1})=*,\\[4pt]
		(\mathbf{B}((\mathcal{T},r,l),\langle k\rangle_{K_1}),x), & \text{otherwise}.
	\end{cases}
\end{equation}
Here $\mathbf{B}(\cdot)$ denotes the corresponding entry of the base TST MIMO-PDA $\mathbf{B}$ defined in \eqref{eq-MN-array-hyger}. For instance, the vector $([6],1)$ of $\mathbf{B}$ in $\mathbf{Q}_1$ becomes $([6],1,1)$, in $\mathbf{Q}_2$ it becomes $([6],1,2)$, and so forth. We can check that $\mathbf{X}$ contains totally $m{t_1+\tau_1\choose  \tau_2}S_1$ different vectors, and each column of  $\mathbf{X}$  has $m{t_1+\tau_1\choose  \tau_2}Z_1$ stars. This is because, each array $\mathbf{Q}_i$ where $i\in [m{t_1+\tau_1\choose  \tau_2}]$ contains $S_1$ different vectors and each of its columns has $Z_1$ stars. Thus Conditions C$1$ and C$2$ of Definition \ref{def-MIMO-PDA} hold. Furthermore, Conditions C$3$ and C$4$ of Definition \ref{def-MIMO-PDA} hold since each array $\mathbf{Q}_i$ satisfies C$3$ and C$4$. So $\mathbf{X}$ is an $g_{\mathbf{X}}=mg_1$-$(L,mK_1$, $F_{\mathbf{X}}=m{t_1+\tau_1\choose  \tau_2}F_1$, $Z_{\mathbf{X}}=m{t_1+\tau_1\choose  \tau_2}Z_1$, $S_{\mathbf{X}}=m{t_1+\tau_1\choose  \tau_2}S_1)$ MIMO-PDAs. 

\subsection{Construction of a MIMO-PDA $\mathbf{Y}$}
\label{sub-sec-array}

To obtain the MIMO-PDA $\mathbf{Y}$ with sum-DoF $g_{\mathbf{Y}}=G(mt_1+\tau)-g_{\mathbf{X}}$, we can reorganize these columns into $(t_1+\tau_1)/\tau_2$ groups each of which has size of $\tau_2$ and generates sum-DoF $g_{\mathbf{Y}}$. 

For convenience in computing the degree of the bipartite graph introduced later, we first replicate $\mathbf{B}$ 
$\Lambda_3 = \binom{t_1+\tau_1-1}{\tau_1} $ times. Then, the columns of each replica are partitioned into $\Lambda_3$ parallel classes of this $t_1+\tau_1$ column set, each consisting of $\tau_2$-subsets that collectively yield the sum-DoF $g_{\mathbf{Y}}$. Specifically, setting $\alpha=\tau_2$ and $v=t_1+\tau_1$ in Lemma~\ref{lem-factor}, every $(t_1+\tau_1)$-subset $\mathcal{S}$ admits $\Lambda_3=\binom{t_1+\tau_1-1}{\tau_2-1}$ parallel classes, which we denote by $\mathcal{U}_{\mathcal{S},1},\dots,\mathcal{U}_{\mathcal{S},\Lambda_3}$ to avoid confusion with the earlier notation $\mathcal{F}_{\mathcal{S},d}$. Define the $\Lambda_3F_1\times K_1$ array
\begin{equation}
	\label{eq-P-T}
	\mathbf{T}=
	\begin{pmatrix}
		\mathbf{B}_1 \\ 
		\mathbf{B}_2 \\ 
		\vdots \\
		\mathbf{B}_{\Lambda_3}
	\end{pmatrix}
\end{equation}
where, for each $y_1\in[\Lambda_3]$, the sub‑array $\mathbf{B}_{y_1}=\bigl(\mathbf{B}_{y_1}((\mathcal{T},r,l),k)\bigr)$ uses the same row index $(\mathcal{T},r,l)$ and column index $k\in[K_1]$ as $\mathbf{B}$. Its entries are
\begin{equation}
	\label{eq-array-B-i}
\mathbf{B}_{y_1}((\mathcal{T},r,l),k)=
	\begin{cases}
		*, & \text{if }\mathbf{B}((\mathcal{T},r,l),k)=*,\\[6pt]
		\bigl(\mathbf{B}((\mathcal{T},r,l),k), y_1, y_2\bigr), & \text{if } k\in\mathcal{U}_{\mathcal{S},y_1}(y_2) \text{ for some } y_2.
	\end{cases}
\end{equation}

Recall that the vector $(\mathcal{S},\lceil\tfrac{d+(l-1)\Lambda_{2}}{G}\rceil)$ in the base array $\mathbf{B}$ occurs in exactly $t_1+\tau_1$ columns, i.e., the column set $\mathcal{S}$, with $G$ identical non-star entries per column. In addition,  $\mathcal{U}_{\mathcal{S},y_1}(y_2)$ denotes the $y_2$-th element of the $y_1$-th parallel of $\mathcal{S}$. This implies that the columns of $\mathcal{S}$ are partitioned into $(t_1+\tau_1)/\tau_2$ disjoint groups each of which has $\tau_2$ columns, and the columns containing the entry $(\mathbf{B}((\mathcal{T},r,l),k), y_1, y_2)$ are exactly those in $\mathcal{U}_{\mathcal{S},y_1}(y_2)$. Therefore, every vector in $\mathbf{B}_{y_1}$ where $y_1\in[\Lambda_{3}]$ satisfies Conditions~C3 and C4 of Definition~\ref{def-MIMO-PDA}, and consequently so does $\mathbf{T}$. Since $\mathbf{B}$ contains exactly $S_1$ distinct vectors, each $\mathbf{B}_{y_1}$ contains $\frac{(t_1+\tau_1)}{\tau_2}S_1$ distinct appended vectors. Hence, $\mathbf{T}$ contains exactly $\frac{(t_1+\tau_1)}{\tau_2}\Lambda_3 S_1$ distinct appended vectors. This implies that $\mathbf{T}$ is a $G\tau_2$-$(G,L'$, $K_1$, $\Lambda_3 F_1$, $\Lambda_3 Z_1$, $\frac{(t_1+\tau_1)}{\tau_2}\Lambda_3 S_1)$ MIMO-PDA.

From \eqref{eq-P-T} and \eqref{eq-array-B-i}, we label each row of $\mathbf{T}$ by the vector $(\mathcal{T},r,l,y_1)$, where $\mathcal{T}\in\binom{[K_1]}{t_1}$, $r\in[\Lambda_{1}]$, $l\in[G/\gcd(G,\Lambda_{2})]$, and $y_1\in[\Lambda_3]$. We can obtain the desired MIMO-PDA $$\mathbf{Y}= (\mathbf{T}_1,\mathbf{T}_2,\ldots,\mathbf{T}_m).$$ For each $y_3\in[m]$, the $\Lambda_3 F_1\times K_1$ array $\mathbf{T}_{y_3}=(\mathbf{T}_{y_3}((\mathcal{T},r,l,y_1),k))$ where $\mathcal{T}\in\binom{[K_1]}{t_1}$, $r\in[\Lambda_{1}]$, $l\in[G/\gcd(G,\Lambda_{2})]$, $y_1\in\bigl[\Lambda_3 \bigr]$, $k\in[K_1]$, and its entries are defined as
\begin{equation}
	\label{eq-array-T-c}
\mathbf{T}_{y_3}((\mathcal{T},r,l,y_1),k)=
	\begin{cases}		
		*& \mathbf{T}((\mathcal{T},r,l,y_1),k)=*,\\
		(\mathbf{T}((\mathcal{T},r,l,y_1),k),y_3),& \text{otherwise},
	\end{cases}
\end{equation}Similarly, we can check that $\mathbf{Y}$ is a $g_{\mathbf{Y}}=G\tau_2$-$(G,L,K=mK_1,\Lambda_3 F_1,\Lambda_3 Z_1,\frac{(t_1+\tau_1)}{\tau_2}m\Lambda_3 S_1)$ MIMO-PDA.

\subsection{Construction of an MAPDA $\mathbf{P}$}
\label{sub-matching}
In this subsection, we will obtain an array $\mathbf{P}$ by concatenating $\mathbf{X}$ and $\mathbf{Y}$ vertically and adjusting the vectors in $\mathbf{Y}$ using the perfect matching Lemma in \cite{BJM}.  

Let $\mathcal{X}$ and $\mathcal{Y}$ denote the sets of distinct vectors appearing in the arrays $\mathbf{X}$ and $\mathbf{Y}$ respectively. From \eqref{eq-array-i}, the vertex in $\mathcal{X}$ can be represented as 
\begin{align}
	\label{eq-x}
	\mathbf{x}=\left(\mathcal{S}, a=\left\lceil\frac{d+(l-1)\Lambda_{2}}{G}\right\rceil,x\right)
\end{align} where $x\in\bigl[m\binom{t_1+\tau_1}{\tau_2}\bigr]$, $d\in[\Lambda_{2}]$, $l\in\bigl[\frac{G}{\gcd(G,\Lambda_{2})}\bigr]$, and $\mathcal{S}\in{[K_1]\choose t_1+\tau_{1}}$. From \eqref{eq-array-B-i}, the vertex in $\mathcal{Y}$ can be represented as 
\begin{align}
	\label{eq-y}
	\mathbf{y}=\left(\mathcal{S}', a'=\left\lceil\frac{d'+(l'-1)\Lambda_{2}}{G}\right\rceil,y_1,y_2,y_3\right)
\end{align} where $y_1\in[\Lambda_3]$, $y_2\in[\frac{t_1+\tau_1}{\tau_1}]$, $d'\in[\Lambda_{2}]$, $l'\in\bigl[\frac{G}{\gcd(G,\Lambda_{2})}\bigr]$, and $\mathcal{S}'\in{[K_1]\choose t_1+\tau_{1}}$.

Recall that the subarray $\mathbf{X}^{(\mathbf{x})}$ of $\mathbf{X}$ has exactly $m(t_1+\tau_1)$ columns, i.e., the column set $\cup_{i=0}^{m-1}(\mathcal{S}+iK_1)$, and each columns has exactly $G$ entries containing $\mathbf{x}$, and the subarray $\mathbf{Y}^{(\mathbf{y})}$ of $\mathbf{Y}$ has exactly $\tau_2$ columns, i.e., the column set $\mathcal{U}_{\mathcal{S}',y_1}(y_2)+(y_3-1)K_1$, and each columns has exactly $G$ entries containing $\mathbf{y}$.  

According to the aforementioned introduction, we can construct a bipartite graph $\mathcal{G}_{\mathbf{X},\mathbf{Y}}=(\mathcal{X},\mathcal{Y};\mathcal{E})$ where a vertex $\mathbf{x}\in\mathcal{X}$ is adjacent with a vertex $\mathbf{y}\in\mathcal{Y}$ if and only if all of the following conditions hold:
\begin{align}
\label{eq-conditions}	
|\mathcal{S}\cup\mathcal{S}'|=t_1+\tau_1+\tau_2, \ \mathcal{S}'\setminus\mathcal{S}=\mathcal{U}_{\mathcal{S}',y_1}(y_2),\ \text{and}\ 
G\leq \langle(a-a')G\rangle_{\Lambda_{2}}\leq \Lambda_{2}-G.
\end{align} For any vertices $\mathbf{x}\in\mathcal{X}$ and $\mathbf{y}\in\mathcal{Y}$, assume that $\mathbf{x}$ is adjacent with $\mathbf{y}$, i.e., they satisfy the conditions in \eqref{eq-conditions}.

Let $\mathbf{P}'$ be the subarray of $\mathbf{P}$ including the row and columns containing the vectors $\mathbf{x}$ and $\mathbf{y}$. Replacing the vertex $\mathbf{y}$ with $\mathbf{x}$, let us consider the modified $\mathbf{P}'$, denoted by $\mathbf{P}^{(\mathbf{x})}$. By the hypothesis $\mathcal{S}\setminus\mathcal{S}=\mathcal{U}_{\mathcal{S}',y_1}(y_2)$, $\mathbf{P}'$ has exactly $m(t_1+\tau_1)+\tau_2$ columns. Recall that each row of 
 $\mathbf{X}^{(\mathbf{x})}$ has exactly $mt_1$ stars entries, i.e., each row has exactly $m\tau_1$ non-star entries. Let $\mathbf{y}=(\mathcal{S}',a',y_1,y_2,y_3)$ and let $\mathbf{Y}^{(\mathcal{S}',y_3)}$ denote the subarray consisting of all rows and columns in which the vector has first coordinate $\mathcal{S}'$ and fourth coordinate $y_3$. Since $\mathbf{Y}$ is generated by the base array $\mathbf{B}$, each row of $\mathbf{Y}^{(\mathcal{S}',y_3)}$ has exactly $t_1$ stars, i.e., each row has exactly $\tau_1$ non-star entries. Clearly $\mathbf{Y}^{(\mathbf{y})}$ is the subarray of $\mathbf{Y}^{(\mathcal{S}',y_3)}$. By the hypothesis $|\mathcal{S}\cup\mathcal{S}'|=t_1+\tau_1+\tau_2$, each row has exactly of $\mathbf{P}'$ has exactly $mt_1$ star entries, i.e., each row has exactly $m\tau_1+\tau_2=\lceil L/G\rceil$ non-star entries. Then the condition C$4$-a) holds.

 Let us consider the condition C$4$ of Definition \ref{def-MIMO-PDA}.  By the definition of $a$ and $a'$, there exit $G$ integers $\mathcal{D}=\{d_1$, $d_2$, $\ldots$, $d_G\}$, and $G$ integers $\mathcal{D}'=\{d_1'$, $d_2'$, $\ldots$, $d_G'\}$ such that for each integer $i\in[G]$ the following equations hold.  
 \begin{align*}
 	a=\left\lceil\frac{d_i+(l-1)\Lambda_2}{G}\right\rceil\ \ \ \text{and}\ \ \ a'=\left\lceil\frac{d_i'+(l'-1)\Lambda_2}{G}\right\rceil.
 \end{align*} Given the parallels of $\mathcal{S}$, the corresponding parallels of $\mathbf{S}'$ can be obtained as follows: the parallel $\mathcal{F}_{\mathcal{S}',d'}$ is obtained by replacing all the integers in $\mathcal{S}$ and not in $\mathcal{S}'$ by the integers in $\mathcal{S}'$ and not in $\mathcal{S}$. When $d\neq d'$, the intersection of the parallels $\mathcal{F}_{\mathcal{S},d}$ and $\mathcal{F}_{\mathcal{S},d'}$ is empty. In this case, the intersection of the parallels $\mathcal{F}_{\mathcal{S},d}$ and $\mathcal{F}_{\mathcal{S}',d'}$ is empty too. In addition, each parallel of $\mathcal{S}$ represent the $\frac{t_1+\tau_1}{\tau_1}$ star position sets of $\mathbf{X}^{(\mathbf{x})}$ and $\mathbf{Y}^{(\mathcal{S},y_3)}$. So, when $\mathcal{D}\cap\mathcal{D}'=\emptyset$, there is no common star position set in $\mathbf{X}^{(\mathbf{x})}$ and $\mathbf{Y}^{(\mathcal{S},y_3)}$. In fact, when 
 $G\leq \langle(a-a')G\rangle_{\Lambda_{2}}\leq \Lambda_{2}-G$ holds, $\mathcal{D}\cap\mathcal{D}'=\emptyset$ always holds. For the detail proof, please see Appendix \ref{app-empty}. This implies that all the rows of $\mathbf{P}'$ has different star position, i.e., $\mu=1$. Then the condition C$4$-b) of Definition \ref{def-MIMO-PDA} holds.

In the following, let us count the degree $\text{d}(\mathbf{x})$ and $\text{d}(\mathbf{y})$ in $\mathcal{G}_{\mathbf{X},\mathbf{Y}}$. It is not difficult to count that there are exactly $${t_1+\tau_1 \choose t_1+\tau_1-\tau_2}{K_1-t_1+\tau_1\choose \tau_2}$$ subsets $\mathcal{S}'$ satisfying $|\mathcal{S}\cup\mathcal{S}'|=t_1+\tau_1+\tau_2$, and $\frac{\Lambda_{2}-2G}{\gcd(\Lambda_{2},G)}+1$ different integers $a'$ in $[\frac{\Lambda_{2}}{\gcd(\Lambda_{2},G)}]$ satisfying the third condition in \eqref{eq-conditions}.
Recall that the integer $y_3$ can run all the values of $[m]$. There are exactly $$\text{d}(\mathbf{x})={t_1+\tau_1 \choose t_1+\tau_1-\tau_2}{K_1-t_1+\tau_1\choose \tau_2}\left(\frac{\Lambda_{2}-2G}{\gcd(\Lambda_{2},G)}+1\right)m$$ vector $\mathbf{y}$ satisfying the conditions in \eqref{eq-conditions}. 

Conversely, it is not difficult to count that there are exactly ${t_1+\tau_1 \choose t_1+\tau_1-\tau_2}{K_1-t_1+\tau_1\choose \tau_2}$ subsets $\mathcal{S}$ satisfying $|\mathcal{S}\cup\mathcal{S}'|=t_1+\tau_1+\tau_2$ and $\frac{\Lambda_{2}-2G}{\gcd(\Lambda_{2},G)}+1$ different integers $a$ in $[\frac{\Lambda_{2}}{\gcd(\Lambda_{2},G)}]$ satisfying the third condition in \eqref{eq-conditions}. Since each $\mathbf{Q}$ is replicated $m\binom{t_1+\tau_1}{\tau_2}$ times to construct the array $\mathbf{X}$ in \eqref{eq-P-Q}, there are exactly
$$\text{d}(\mathbf{x})={t_1+\tau_1 \choose t_1+\tau_1-\tau_2}{K_1-t_1+\tau_1\choose \tau_2}\left(\frac{\Lambda_{2}-2G}{\gcd(\Lambda_{2},G)}+1\right)m\binom{t_1+\tau_1}{\tau_2}$$ 
vector $\mathbf{x}$ satisfying the conditions in \eqref{eq-conditions}. Since $\mathbf{x}$ and $\mathbf{y}$ is any vertices of $\mathcal{X}$ and $\mathcal{Y}$ respectively, we have d$(\mathcal{X})=$d$(\mathbf{x})$ and d$(\mathcal{Y})=$d$(\mathbf{y})$. Since $\text{d}(\mathcal{X})\leq \text{d}(\mathbf{Y})$, by Lemma \ref{lem-factor} there exists a matching with $|\mathcal{X}|$ edges. Since $|\mathcal{X}|=|\mathcal{Y}|$, the bipartite graph $\mathcal{G}_{\mathbf{X},\mathbf{Y}}$ has a perfect matching.

\section{The proof of $\mathcal{D}\cap\mathcal{D}'=\emptyset$}
\label{app-empty}From \eqref{eq-MN-array-hyger}, we have 
\begin{align*}
\mathcal{D}&=\{\langle(a-1)G+1\rangle_{\Lambda_{2}},\langle(a-1)G+2\rangle_{\Lambda_{2}},\ldots,\langle aG\rangle_{\Lambda_{2}}\}\\
\mathcal{D}'&=\{\langle(a'-1)G+1\rangle_{\Lambda_{2}},\langle(a'-1)G+2\rangle_{\Lambda_{2}},\ldots,\langle a'G\rangle_{\Lambda_{2}}\}
\end{align*}Assume that there exist two integers $i$, $j\in[G]$, satisfying $\langle(a-1)G+i\rangle_{\Lambda_{2}}=\langle(a'-1)G+j\rangle_{\Lambda_{2}}$ which implies that
\begin{align}
\label{eq-eqaulity}
\langle(a'-a)G\rangle_{\Lambda_{2}}=\langle(i-j)G\rangle_{\Lambda_{2}}.
\end{align}Since $i - j \in [-(G-1),G-1]$, \eqref{eq-eqaulity} holds if and only if $\langle(a'-a)G\rangle_{\Lambda_{2}}\in [-(G-1),G-1]$. This impossible since 
$G\leq \langle(a-a')G\rangle_{\Lambda_{2}}\leq \Lambda_{2}-G$.

\end{appendices}

\bibliographystyle{IEEEtran}
\bibliography{reference}

\end{document}